\documentclass[final,3p,12pt,authoryear,a4paper]{elsarticle}

\usepackage{float}
\usepackage[none]{hyphenat}
\usepackage{amsmath,amsfonts,amssymb,amsthm,pxfonts,dsfont}
\usepackage{natbib}
\usepackage[bookmarks, colorlinks, pdftitle={Time-Varying Identification of Monetary Policy Shocks},pdfauthor={Annika Camehl, Tomasz Wozniak}]{hyperref}
\hypersetup{linkcolor=blue,citecolor=blue,filecolor=black,urlcolor=blue} 
\usepackage{graphicx} 
\usepackage{color}
\usepackage{geometry}
\usepackage{multirow}
\usepackage{booktabs}
\usepackage{rotating} 
\usepackage{longtable}
\usepackage{subfig}

\usepackage[american]{babel}

\emergencystretch=\maxdimen
\usepackage{setspace}
\DeclareMathOperator{\diag}{diag}

\begin{document}

\biboptions{longnamesfirst}
\newtheorem{thm}{Theorem}
\newtheorem{cl}{Corollary}
\newtheorem{prop}{Proposition}
\newtheorem{lm}{Lemma}
\newdefinition{df}{Definition}
\newdefinition{ex}{Example}
\newdefinition{as}{Assumption}
\newdefinition{pr}{Property}
\newdefinition{rs}{Restriction}
\newdefinition{al}{Algorithm}

\definecolor{myb}{rgb}{0.13, 0.67, 0.8}
\journal{}

\begin{frontmatter}

\title{Time-Varying Identification of Monetary Policy Shocks}


\author[eur]{Annika Camehl\corref{cor1}}
\author[um]{Tomasz Wo\'zniak\corref{cor2}}

\cortext[cor1]{\textit{Corresponding author}: Annika Camehl, Erasmus University Rotterdam, Department of Econometrics, Burgemeester Oudlaan 50, 3062 PA Rotterdam, The Netherlands, \textit{email:} \href{mailto:camehl@ese.eur.nl}{camehl@ese.eur.nl}.\\[1ex]
\copyright{} 2024 Annika Camehl \& Tomasz Wo\'zniak}
%
\address[eur]{Erasmus University Rotterdam}
\address[um]{University of Melbourne}

\begin{abstract}
We propose a new Bayesian heteroskedastic Markov-switching structural vector autoregression with data-driven time-varying identification. 
The model selects alternative exclusion restrictions within regimes and, as a condition for the search, allows to verify identification through heteroskedasticity. 
We show that US data support time variation in US monetary policy shock identification.
In the sample-dominating first regime, systematic monetary policy follows a Taylor rule extended by the term spread, effectively curbing inflation. 
In the second regime, gaining more persistence after the global financial and COVID crises, it is characterized by a money-augmented Taylor rule, providing economic stimulus, and featuring the liquidity effect.
\end{abstract}

\begin{keyword}\normalsize
Structural VARs \sep Markov Switching\sep	Identification Via Heteroskedasticity \sep Extended Taylor Rule \sep Effects of Monetary Policy \\
	\textit{JEL classification:} 
	C11, 
	C32, 
	E52 
\end{keyword}

\end{frontmatter}

\newpage
\section{Introduction}\label{sec:Introduction}

\noindent The targeting strategy of the Federal Reserve (Fed) to reach its dual mandate of price stability and economic output has evolved over the last sixty years. The Fed moved in the 1980s from acting via adjustments of the money supply towards inflation-targeting interest rate policy, adopted a new unconventional toolkit when the federal funds rate reached its effective zero lower bound in December 2008, and faced changes in the regulatory environment or the importance of the financial sector. Despite these policy changes, empirical studies on the effects of monetary policy analyzing variables' responses to exogenous shifts to the monetary policy rule in structural vector autoregressions assume that restrictions used to identify the structural shocks do not change over time following the setup proposed by \cite{sims1980}. This time-invariance in identification is imposed in the literature regardless of the identification strategy implemented or time variation in the parameter values \citep[see, e.g.,][]{Primiceri2005,Sims2006}. However, such time-invariant identification strategies might be at odds with data and reality. We relax this time-invariance assumption implementing time-varying identification via a~data-driven search mechanism.

To that end, we develop a novel model, a Bayesian Markov-switching (MS) structural vector autoregression (SVAR) with time-varying identification (TVI) and heteroskedasticity following stochastic volatility (SV). It automatically selects from alternative patterns of exclusion restrictions identifying monetary policy shocks via a~data-driven search process within each MS regime. We introduce a new prior distribution to estimate a TVI indicator altering the identification pattern in each regime. Consequently, the structural parameters might take different values and different exclusion restrictions might be selected in each of the MS states. No identifying pattern is associated with a specific regime. Hence, we do not exclude the possibility that identification of the monetary policy shock is time-invariant even if structural parameters vary with time. Instead, we let the data provide evidence in our search mechanism on which exclusion restrictions are supported within a regime. 

Our regime-dependent selection from a range of alternative exclusion restrictions is valid thanks to identification via heteroskedasticity. Our structural matrix follows a slow-moving MS process. The structural shocks' variances instead evolve at a higher frequency according to an SV process, which is controlled by a volatility of the log-volatility parameter following MS. This higher frequency enables that a shock is heteroskedastic within a regime. Thereby, the shock can be identified via heteroskedasticity within a regime rendering the exclusion restrictions overidentifying, which allows discriminating among these restrictions in our data-driven search. Additionally, the identifying restrictions provide economic labelling of the shocks identified via heteroskedasticity. We verify identification via heteroskedasticity in each regime by checking that the regime-dependent volatility of the log-volatility parameter does not imply homoskedastic shocks within a regime. Finally, we boost the model's flexibility by applying a three-level equation-specific hierarchical prior distribution to the autoregressive and structural parameters, enabling prior shrinkage estimation.

We select the monetary policy shock identifying patterns from alternatives motivated by theoretical and empirical SVAR literature. The first such pattern is the Taylor Rule proposed by \cite{Taylor1993}, in which interest rates are related to inflation and output fluctuations, and which was employed in the seminal SVAR paper by \cite{leeperWhatDoesMonetary1996a}. Another two monetary policy rules extend the Taylor Rule by the term spread as in \cite{Vazquez2013} and \cite{diebold2006}, or money as suggested by \cite{Andres2006} and \cite{BelongiaIreland2015}. As the last alternative in our benchmark model, we consider a simple money-interest rate rule used in \cite{LeeperRoush2003} and \cite{Sims2006}. We also report results where the Taylor Rule is extended by credit spread as in \cite{Curdia2010}, or stock prices as in \cite{Fuhrer2008}. 

For each regime, we estimate the probability for all rules by our TVI mechanism. In order to verify which of the policy rules are supported by the data, we use six US monthly time series from January 1960 to June 2023, including industrial production, consumer price index inflation, federal funds rate, term spread, M3 money, and stock prices. We allow for two MS states as our data do not favor a larger number of regimes.

We show that data strongly support time variation in US monetary policy shock identification. In the first regime, which dominates before 2000 and in more tranquil economic conditions, data robustly opt for including term spreads in the Taylor Rule in the spirit of \cite{diebold2006}. This suggests that conventional monetary policy is implemented giving consideration to the slope of the bond yield curve. It is also effective in curbing inflation, including that on the stock exchange, and features the liquidity effect. An expansionary monetary policy shock flattens the bond yield curve in this regime. 

In contrast, in the second regime occurring mainly after 2000, predominately during the global financial crisis (GFC) and COVID-19 pandemic and during periods of unconventional monetary policy actions, data support a monetary policy rule augmented by the monetary aggregate as advocated by \cite{BelongiaIreland2015}. A~monetary policy shock lowering interest rates provides a substantial stimulus to  economic activity, boosts money supply much stronger than in the first regime, and steepens the bond yield curve. Additionally, it is complemented by a pure term spread shock that moves the term spread not affecting interest rates simultaneously as defined by \cite{Baumeister2013}. We show that if monetary policy as implemented in the second regime did not occur, inflation would be on average 1 percentage point (pp) higher and the interest rate would follow the shadow rate path after 2008.

Our results clearly indicate that monetary policy provided a cushioning effect in the aftermath of sharp crises by boosting economic activity. This came with over half-year delay and at a cost of a slight increase in inflation. Our interpretations would not hold in the absence of TVI which sharpens identification and addresses model misspecifications that lead to the price puzzle. We provide evidence for partial identification of the monetary policy shock via heteroskedasticity validating the regime-dependent TVI.

To the best of our knowledge, we are the first to introduce time-varying identification via a data-driven search. Only a few previous papers adopt regime dependent identification patterns, albeit setting restrictions \textit{a priori} without further validation. 
\cite{KimuraNakajima2016} analyze changes in Japan's monetary policy using two exogenously given regimes, an interest rate or bank reserves rule, embedded within a latent threshold time-varying parameter SVAR.
\cite{Bacchiocchi2017} impose two identification patterns with an additive relationship in a heteroskedastic SVAR and apply them to study the effects of monetary policy. They report differences in impulse responses during the Great Moderation compared to previous periods. 
\cite{Arias2024} impose sign restrictions on impulse responses to identify US monetary policy shocks and add restrictions on contemporaneous structural parameters in \textit{a priori}-defined periods when the Fed used the federal funds rates as the main policy instrument. With our search mechanism, we challenge the common assumption of time-invariant US monetary policy shock identification and we estimate the probabilities of regime occurrences. Moreover, we contribute to the discussion on appropriate exclusion restrictions to identify US monetary policy shocks by including a range of various patterns in our search mechanism. 

In the following, we introduce our new model in Section \ref{sec:Model}, report the empirical evidence on time-varying identification in Section \ref{sec:resultsTVI}, and discuss the effects of US monetary policy shocks in Section \ref{sec:resultsIRF}.

\section{Time-varying identification in SVARs}\label{sec:Model}

\noindent In this section, we introduce our heteroskedastic SVAR with MS time variation and TVI in the structural matrix, and the four monetary policy rules TVI selects from in each regime. We discuss the new prior distribution facilitating the data-dependent search of time-varying exclusion restrictions. Finally, we explain identification via heteroskedasticity, based on which we achieve that the exclusion restrictions are over-identifying, and the priors for all parameter groups.

\subsection{Model specification}\label{subsec:Modelspecification}

\noindent To study time variation in the identification of US monetary policy shocks, we estimate a structural vector autoregressive model:
\begin{align}
	\mathbf{y}_t &= \sum_{l=1}^{p} \mathbf{A}_l \mathbf{y}_{t-l} + \mathbf{A}_d \mathbf{d}_t + \boldsymbol{\varepsilon}_t  \label{eq:rf} \\
\mathbf{B}\left(s_t, \boldsymbol{\kappa}(s_t)\right) \boldsymbol{\varepsilon}_t &= \mathbf{u}_t \label{eq:sf}
\end{align}
where the vector $\mathbf{y}_t$ of $N=6$ dependent variables at time $t$ follows a vector autoregressive model in equation~\eqref{eq:rf} with $p=6$ lags, an $d=1$ deterministic term $\mathbf{d}_t$ that we set to a constant term, an error term $\boldsymbol{\varepsilon}_t$, and $N\times N$ autoregressive matrices $\mathbf{A}_l$, for lag $l = 1, \dots, p$, and the $N\times d$ matrix $\mathbf{A}_d$ of deterministic term slopes. The structural equation~\eqref{eq:sf} links the reduced form residuals to the structural shocks $\mathbf{u}_t$ in a time-varying linear relationship via the $N\times N$ structural matrix $\mathbf{B}\left(s_t, \boldsymbol{\kappa}(s_t)\right)$.

The system contains monthly data on industrial production, denoted by $y_t$, consumer price index inflation, $\pi_t$, federal funds rates, $R_t$, term spread, $TS_t$, measured as the 10-year treasury constant maturity rate minus the federal funds rate, M2 money supply, $m_t$,  and the S\&P500 stock price index, $sp_t$. 
$R_t$ and $TS_t$ are expressed in annual terms, as is $\pi_t$ computed as the logarithmic rates of returns on the consumer price index expressed in percent. $y_t$, $m_t$, and $sp_t$ are taken in log-levels multiplied by 100.\footnote{
Series $R_t$, $y_t$, and  $m_t$ are downloaded from \cite{FFR,IP,M} respectively, $\pi_t$ is provided by \cite{CPI}, $TS_t$, by \cite{TS}, and $sp_t$ by \cite{SP}. Additionally, we used a series of shadow rates by \cite{Wu2016} downloaded from \url{https://sites.google.com/view/jingcynthiawu/shadow-rates} [accessed 15.08.2023].}
These variables allow us to specify various identification patterns for monetary policy shocks and to study the impact of unexpected monetary expansion on the economic environment, monetary measures, and financial indicators. 

The crucial feature of our model is the relationship between the reduced- and structural-form shocks in equation~\eqref{eq:sf} captured by the structural matrix that is regime-dependent and, hence, the additional restrictions to identify structural shocks imposed on the contemporaneous relations as well as the parameter estimates can vary over time. The structural matrix  $\mathbf{B}\left(s_t, \boldsymbol{\kappa}(s_t)\right)$ depends on a regime indicator $s_t$ of a discrete Markov process with $M$ states and on a regime-specific collection of TVI indicators, $\boldsymbol{\kappa}(s_t)$. The TVI indicator selects amongst $K$ patterns of the exclusion restrictions imposed on the monetary policy equation for each regime. Thus, we can select structural shocks' identification patterns which are specific to the regimes identified by the MS process. We do not associate \textit{a priori} an identifying pattern with a specific regime but estimate for each of the four exclusion restriction schemes the occurrence probability within each regime. In effect, a MS model without TVI might be estimated and the the TVI will only occur if data supports it. Moreover, irrespective of TVI, MS in the structural coefficients enables that impulse responses are time-varying as in \cite{hamilton2016}, which  \cite{LutkepohSchlaak2022} find to be important in heteroskedastic models.  

We assume that the MS process, $s_t$, by \cite{hamilton1989} is stationary, aperiodic, and irreducible, with a $M\times M$ transition matrix $\mathbf{P}$ and an $N$-vector of initial values $\boldsymbol{\pi}_0$, denoted by $s_t \sim\mathcal{M}arkov(\mathbf{P}, \boldsymbol{\pi}_0)$.

The data-driven selection among the exclusion restriction patterns is possible if these are overidenifying. We achieve this by identification via heteroskedaticity. To that end, our model includes heteroskedastic structural shocks that follow SV. We specify that the structural shocks at time $t$ are contemporaneously and temporarily uncorrelated and jointly conditionally normally distributed given the past observations on vector $\mathbf{y}_t$, denoted by $\mathbf{Y}_{t-1}$, with zero mean and a diagonal covariance matrix:
\begin{align}
\mathbf{u}_t \mid \mathbf{Y}_{t-1} &\sim\mathcal{N}_T\left(\mathbf{0}_N, \diag\left(\boldsymbol{\sigma}_t^2\right)\right) \label{eq:sfshock}
\end{align}
where $\boldsymbol{\sigma}_t^2$ is an $N$-vector of structural shocks' conditional variances at time $t$ filling in the main diagonal of the covariance matrix, and $\mathbf{0}_N$ is a vector of $N$ zeros. The third element in $\mathbf{u}_t$ is the monetary policy shock that is heteroskedastic.

\subsection{Four monetary policy rules}\label{sec:mpr}

\noindent Our selection of identifying patterns chooses from four ($K=4$) monetary policy rules specified for the third equation and, thus, the third row of the structural matrix in each of the MS regimes. Our restriction patterns embed four policy rule specifications, shown in Table~\ref{tab:Bidentrest}. Exclusion restrictions are widely used to identify structural shocks due to their ease of use \citep[e.g.,][impose various exclusion restrictions to identify the monetary policy shock]{leeperWhatDoesMonetary1996a,Primiceri2005,Sims2006,Wu2016}. In the context of statistical identification procedures, such as via heteroskedasticity or non-normality, exclusion restrictions provide the necessary economic assumptions to label the structural shocks assigning them specific interpretation.

\begin{table}[h!] 
	\caption{Monetary policy shocks identification patterns used in  regime-specific time-varying identification}\label{tab:Bidentrest} \centering
 
	\begin{tabular}{lrcccccc} 
	&	&$y_t$ & $\pi_t$ & $R_t$ & $TS_t$ & $m_t$ &  $sp_t$ \\\toprule	  
	\textit{TR} &	Taylor rule &$*$ &$*$ & $*$ & 0 & 0 & 0  \\\cline{3-8}
	\textit{TR with TS} &	Taylor rule with term spread &$*$ &$*$ & $*$ & $*$ & 0 & 0  \\\cline{3-8}
	\textit{TR with m} &	Taylor rule with money &$*$ &$*$ & $*$ & 0 &$*$ & 0  \\\cline{3-8}	
	\textit{MIR} &	Money-interest rate rule &0 &0 & $*$ & 0 &$*$ & 0  \\\bottomrule
	\end{tabular}\\[0.2cm]	
\begin{flushleft}
		{\footnotesize Note: $*$ indicates an unrestricted parameter and $0$ an exclusion restriction }
\end{flushleft}
\end{table}

We denote the monetary policy reaction function depicting the standard empirical Taylor rule (TR) \citep{Taylor1993}, in which short-term interest rates react contemporaneously only to movements in prices and output, by \textit{TR}. 
Another identification pattern that extends the TR by the term spread, $TS_t$ is labeled \textit{TR with TS}. 
We consider this policy rule in response to the monetary policy shock identification strategies applied by \cite{Baumeister2013}, \cite{Feldkircher2016}, and \cite{Liu2017} but emphasise the spread's role in determining the interest rates.
The additional estimated parameter captures the yield curve's slope effect  and explores participants' future expectations in the bond market \citep{Nimark2008,Tillmann2020}. 
Moreover, a~shift in the term spread is an early indicator of business cycle fluctuations to which the central bank should react by adjusting interest rates \citep[see][]{Rudebusch2008,Ang2011,Vazquez2013}.

A version of the TR that contains no exclusion restriction on the contemporaneous reaction to money is inspired by \citep{BelongiaIreland2015} and labeled \textit{TR with m}.
Here, the data reflects interest rate setting which targets money and can arise due to unconventional monetary policy actions, such as quantitative easing, that have been shown to lead to a systematic expansion of the money supply. 
Including money in the monetary policy reaction function suggests simultaneity in interest rates and market liquidity levels. 
Finally, \cite{BelongiaIreland2015} and \cite{LutkepohlWozniak2017} found empirical support for money in policy rules in SVARs without time variation in the parameters.

The fourth identification pattern for the monetary policy reaction function depicts a money-interest rate rule and is labeled \textit{MIR}. It highlights the simultaneity between short-term interest rates and a monetary aggregate. Allowing money to enter the monetary policy reaction function contemporaneously is in line with the identification used by \cite{LeeperRoush2003} and  \cite{Sims2006}. 

Examples of such generalized TRs can also be found in dynamic general equilibrium or New Keynesian models. In those theoretical models, the monetary authority changes interest rates additionally in response to term spreads \citep{Vazquez2013}, growth in money \citep{Andres2006}, or changes in credit spread and aggregated credit \citep{Curdia2010}, or stock prices \citep{Fuhrer2008}. Likewise, \cite{Gertler2011} characterize monetary policy by a standard Taylor rule in normal times, but when the credit spread rises sharply, they allow the central bank to react to credit policy or implement large-scale asset purchases. Related to these generalized TR specifications, theoretical and empirical multi-country models encompass monetary policy reaction functions augmented by exchange rate movements \citep[e.g.,][]{LubikSchorfheide2007,Camehl2023}.  

We impose a lower-triangular structure on the remaining rows of the structural matrix, which identifies them based on the exclusion restrictions only. Hence, for these rows we do not require to verify identification via heteroskedasticity. Note that as long as patterns \textit{TR} and \textit{MIR} lead to identified systems, monetary rules \textit{TR with TS} and \textit{TR with m} do not satisfy the conditions for global or partial identification using exclusion restrictions due to insufficient number of imposed restrictions.   

\subsection{Time-varying identification}\label{sec:TVI}

\noindent To search for changes in the identification pattern of US monetary policy shocks across regimes, we introduce the TVI mechanism and implement it via a new prior distribution. For simplicity, denote by $\mathbf{B}_{m.k}$ the matrix $\mathbf{B}\left(s_t, \boldsymbol{\kappa}(s_t)\right)$ for given realisations of the MS process $s_t=m$, where $m$ denotes one of the $M$ regimes, and the TVI indicator $\boldsymbol{\kappa}(m)=k$, where $k$ stands for one out of four exclusion restrictions pattern. Following \cite{WaggonerZha2003}, we decompose the monetary policy reaction function, that is, the third row of the structural matrix, $[\mathbf{B}_{m.k}]_{n\cdot}$ for $n=3$, into a $1\times r_{n.m.k}$ vector $\mathbf{b}_{n.m.k}$, collecting the unrestricted elements to be estimated, and an $r_{n.m.k}\times N$ matrix $\mathbf{V}_{n.m.k}$, containing zeros and ones placing the elements of $\mathbf{b}_{n.m.k}$ at the appropriate spots, $[\mathbf{B}_{m.k}]_{n\cdot} = \mathbf{b}_{n.m.k} \mathbf{V}_{n.m.k}$,
where $r_{n.m.k}$ is the number of parameters to be estimated for the restriction pattern $k$. Hence, the restriction matrix $\mathbf{V}_{n.m.k}$ contains the information on which parameters are estimated and on which an exclusion restriction is imposed for regime $m$ and restriction pattern $k$. We pre-specify the four different restriction patterns on the monetary policy rule to identify the monetary policy shock as given in Table~\ref{tab:Bidentrest} while imposing a lower-triangular structure on the remaining rows of the structural matrix.

Statistical selection from the four monetary policy rules in each of the regimes is facilitated by a hierarchical prior distribution. Given the fixed regime $s_t = m$ and TVI component indicator, $\boldsymbol{\kappa}(m)=k$, we set a 3-level local-global hierarchical prior on the unrestricted elements $\mathbf{b}_{n.m.k}$. We assume that  $\mathbf{b}_{n.m.k}$ is zero-mean normally distributed with an estimated shrinkage. This equation-specific level of shrinkage for $\gamma_{B.n}$ follows an inverted gamma 2 distribution with scale $\underline{s}_{B.n}$ and shape $\underline{\nu}_B$. The former hyper-parameter has a local-global hierarchical gamma-inverse gamma 2 prior, with the global level of shrinkage determined by $\underline{s}_{\gamma_B}$:
\begin{align}
\mathbf{b}_{n.m.k}' &\mid \gamma_{B.n}, \boldsymbol{\kappa}(m)=k \sim\mathcal{N}_{r_{n.m.k}}\left(\mathbf{0}_{r_{n.m.k}}, \gamma_{B.n} \mathbf{I}_{r_{n.m.k}}\right),  \qquad \text{with}\label{eq:priorbb}\\
\gamma_{B.n}\mid \underline{s}_{B.n} &\sim\mathcal{IG}2\left( \underline{s}_{B.n}, \underline{\nu}_B \right),  \underline{s}_{B.n} \mid \underline{s}_{\gamma_B} \sim\mathcal{G}\left( \underline{s}_{\gamma_B}, \underline{\nu}_{\gamma_B} \right),  \underline{s}_{\gamma_B} \sim\mathcal{IG}2\left( \underline{s}_{s_B}, \underline{\nu}_{s_B} \right),\label{eq:priorb}
\end{align} 
where $\mathbf{I}_{r_{n.m.k}}$ is the identity matrix. The subscript $n$ denotes an equation-specific parameter. The equation-specific hierarchical prior guarantees high flexibility and avoids arbitrary choices as the level of shrinkage is estimated within the model. We set $\underline{\nu}_B=10, \underline{\nu}_{\gamma_B}=10 ,  \underline{s}_{s_B}=100,$ and $\underline{\nu}_{s_B}= 1$ to allow a wide range of values for the structural matrix elements and show extraordinary robustness of our results that applies as long as the shrinkage towards the zero prior mean is not too imposing. This conditional normal prior is equivalent to that by \cite{WaggonerZha2003} for a time-invariant model and to a variant of the generalized--normal prior of \cite{arias2018inference}.  We combine \eqref{eq:priorbb} with a multinomial prior distribution for the TVI component indicator $\boldsymbol{\kappa}(m) = k \in\{1,\dots,K\}$ with flat probabilities equal to $\frac{1}{K}$, denoted by $\boldsymbol{\kappa}(m) \sim\mathcal{M}ultinomial\left(K^{-1},\dots,K^{-1}\right)$.
This prior setup allows us to estimate posterior probabilities for each TVI component in each regime and provides data-driven evidence in favour of the underlying identification pattern. The indiscriminate multinomial prior reflects our agnostic view of which policy rule applies.

Our TVI prior specification is a multi-component generalisation of \cite{geweke1996variable}'s spike-and-slab prior. It is closely related to the stochastic search variable selection \citep[see the seminal paper by][]{George1997}, which we generalize by allowing for multiple components and by the treatment of the restricted component. In line with a spike-and-slap prior, we can write the prior distribution for an $(n,i)^{\text{th}}$ element of the regime-specific matrix $\mathbf{B}(m, \boldsymbol{\kappa}(m))$, denoted by $[\mathbf{B}_{m}]_{n.i}$, marginalized over $\boldsymbol{\kappa}(m)$ as
\begin{align}
	[\mathbf{B}_{m}]_{n.i} \mid \gamma_{B.n} &\sim \frac{K - K_{R}}{K}\mathcal{N}(0, \gamma_B) + \frac{K_{R}}{K}\delta_0.
\end{align}
The element $[\mathbf{B}_{m}]_{n.i}$ is restricted to zero in $K_{R}$ components, whereas it stays unrestricted and normally distributed in $K-K_{R}$ of them. Therefore, its prior distribution is a Dirac mass at value zero, $\delta_0$, with probability $\frac{K_{R}}{K}$, and normal with probability $\frac{K - K_{R}}{K}$. 

To estimate the parameters of the structural model, we use a Gibbs sampler.  Having obtained $S$ posterior draws, we can easily calculate which set of exclusion restrictions is best supported by the data by estimating the posterior probability for each restriction pattern:
\begin{align}
	\widehat\Pr\left[\kappa(m)= k\mid \mathbf{Y}_T\right] = S^{-1}\sum_{s=1}^{S}\mathcal{I}(\kappa(m)^{(s)}=k).
\end{align} 
The probability is estimated by a fraction of posterior draws of the TVI component indicator, $\kappa(m)^{(s)}$, for which it takes a particular value $k$, for each of its values from 1 to $K$. 

In general, the TVI mechanism can also be applied to multiple rows in the structural matrix and, hence, can search for TVI of multiple structural shocks making it adaptable to a broad set of applications.  It is then straightforward to calculate the conditional posterior distribution of the $n\textsuperscript{th}$ equation specification given a particular specification of the $i\textsuperscript{th}$ row. Also, if a specific economic interpretation is given to a particular combination of exclusion restrictions for the structural matrix, one can estimate the joint posterior probability of such a specification by computing the fraction of the posterior draws for which this combination of TVI indicators holds.

\subsection{Identification via heteroskedasticity} \label{sec:identheteroskedasticity}

\noindent Our search of the identification pattern is feasible only if data can distinguish between the monetary policy rules. Our TVI mechanism in the third row of the structural matrix requires that the third structural shock, which we label as the monetary policy shock based on the exclusion restrictions, is identified within each regime, absent the potential additional restrictions we are searching for. 
We partially identify the monetary policy shock and the third row of the structural matrix within a regime via heteroskedasticity. The partial identification via heteroskedasticity focuses on the identification of a specific shock and all parameters of the corresponding row of the structural matrix within a regime without the necessity of imposing additional restrictions.
The feasibility of identification via heteroskedasticity within regimes relies on the volatility of the structural shocks evolving at a higher frequency than the structural matrix following MS with persistent states -- an essential feature of our model. Once the identification of this shock via heteroskedasticity is verified in each of the regimes, any set of exclusion restrictions overidentify the shock enabling the data to discriminate between the alternative policy rules.

Partial identification of the $n^{\text{th}}$ row of the regime-specific structural matrix, denoted by $\mathbf{B}_m$, where the dependence on $k$ is neglected, within each regime is obtained under the following corollary. 
\begin{cl} \label{cl1}
Consider the SVAR model with the structural matrix following MS from equations \eqref{eq:rf}--\eqref{eq:sfshock}. Denote by $t_1^{(m)}, \dots, t_{T_M}^{(m)}$ the time subscripts of the $T_m$ observations allocated to the Markov-switching regime $s_t = m \in \{1,\dots,M\}$. Let $\boldsymbol{\sigma}_{n.m}^2 = (\sigma_{n.t_1^{(m)}}^2,\dots,\sigma_{n.t_{T_M}^{(m)}}^2)$ be the vector containing all conditional variances $\sigma_{n.t}^2$ of the $n^{\text{th}}$ row associated with the observations in the $m\textsuperscript{th}$ regime. Then the $n^{\text{th}}$  row of $\mathbf{B}_m$ is identified up to a sign if 
$\boldsymbol{\sigma}_{n.m}^2 \neq \boldsymbol{\sigma}_{i.m}^2$ for all $i \in \{1, \dots, N\}\setminus \{n\}$.
\end{cl} 
\begin{proof}
The proof proceeds by the same matrix result as in Corollary 1 by \cite{LSUW2022} set for a heteroskedastic SVAR model without time-variation in the structural matrix. In line with Theorem 1 by \cite{LutkepohlWozniak2017}, the matrix result holds even if the inequality above is true for two periods.
\end{proof}

According to Corollary~\ref{cl1}, the $n^{\text{th}}$ row of $\mathbf{B}_{m.\mathbf{k}}$ is identified through heteroskedasticity in the $m\textsuperscript{th}$ regime if changes in the conditional variances of the $n\textsuperscript{th}$ structural shocks evolve non-proportionally to those of other shocks within that regime. 
Global identification of the $\mathbf{B}_m$ via heteroskedasticity is achieved when Corollary \ref{cl1} holds for all  $n \in \{1, \dots, N\}$. Finally, our identification implies the posterior distribution depends on a rotation matrix that changes the sign of structural matrix rows \citep[see][]{bh2015,hwz2007}. We resolve this by applying the normalisation of the signs by \citep[][]{wz2003a}.

In order to swiftly check whether the structural shock is identified via heteroskedasticity within each regime, we introduce as a second new feature of our model that the SV process of the structural shocks has a regime-dependent volatility parameter. Each of the $N$ conditional variances, $\sigma_{n.t}^2$, follows a non-centered SV process with regime-dependent volatility of the log-volatility that decomposes the conditional variances into 
\begin{align} \label{eq:SV}
	\sigma_{n.t}^2 = \exp\{\omega_n(s_t) h_{n.t}\} \quad\text{with}\quad  h_{n.t} = \rho_n h_{n.t-1} + v_{n.t} \quad\text{and}\quad v_{n.t} \sim \mathcal{N}(0,1)
\end{align}
where $\omega_n(s_t)$ is the volatility of the log-volatility parameter defined as plus-minus square-root of the conditional variance of the log-conditional variances $\log\sigma_{n.t}^2$. The log-volatility of the $n^{\text{th}}$ structural shock at time $t$,  $h_{n.t}$, follows an autoregressive model with initial value $h_{n.0} = 0$, where $\rho_n$ is the autoregressive parameter and $v_{n.t}$ is a standard normal innovation. The specification extends the SV process by \cite{Kastner2014} and \cite{LSUW2022} by the regime-dependence in the parameter $\omega_n(s_t)$. Therefore, the Markov process $s_t$ drives the time variation in both $\mathbf{B}\left(s_t,  \boldsymbol{\kappa}(s_t)\right)$ and $\omega_n(s_t)$.  

This specification allows us to verify the hypothesis of within-regime homoskedasticity of each of the shocks by checking the restriction \citep[see][for time-invariant heteroskedastic models]{Chan2018}
\begin{align}
\omega_{n}(m) = 0. \label{eq:omegarestric}
\end{align}
If it holds, the $n\textsuperscript{th}$ structural shock is homoskedastic in the $m\textsuperscript{th}$ regime. If it does not, then the structural shock's conditional variance changes non-proportionally to those of other shocks with probability 1, rendering the shock identified via heteroskedasticity. Consequently, two cases guarantee the identification of the monetary policy shock. First, this shock, and no other, is homoskedastic within a specific regime. Hence, the condition in equation~\eqref{eq:omegarestric} only holds for the third shock. Second, the third structural shock is heteroskedastic in that regime implying that equation~\eqref{eq:omegarestric} is not satisfied for the monetary policy shock.

\subsection{Priors, Standardisation, and Posterior Sampler}\label{sec:priorandposterior}

\noindent In this section, we discuss the prior specification of the remaining parameters of the model and their estimation. We use the following notation: $\mathbf{e}_{n.N}$ is the $n\textsuperscript{th}$ column of the identity matrix of order $N$, $\boldsymbol{\imath}_n$ is a vector of $n$ ones, $\mathbf{p}$ is a vector of integers from 1 to $p$, and $\otimes$ is the Kronecker product.

Let the $N\times (Np+d)$ matrix  $\mathbf{A}=\begin{bmatrix}\mathbf{A}_1 &\dots & \mathbf{A}_p & \mathbf{A}_d\end{bmatrix}$ collect all autoregressive and constant parameters. Extending the prior by \cite{Giannone2015}, each row of matrix $\mathbf{A}$ follows an independent conditional multivariate normal distribution with a 3-level global-local hierarchical prior on the equation-specific shrinkage parameters $\gamma_{A.n}$:
\begin{align}
[\mathbf{A}]_{n\cdot}' &\mid \gamma_{A.n} \sim\mathcal{N}_{Np+d}\left(\underline{\mathbf{m}}_{A.n}', \gamma_{A.n} \underline{\boldsymbol{\Omega}}_A\right),  \qquad \text{with}\\
\gamma_{A.n}\mid \underline{s}_{A.n} &\sim\mathcal{IG}2\left( \underline{s}_{A.n}, \underline{\nu}_A \right), \underline{s}_{A.n} \mid \underline{s}_{\gamma_{A.n}} \sim\mathcal{G}\left( \underline{s}_{\gamma_A}, \underline{\nu}_{\gamma_A} \right), \underline{s}_{\gamma_A} \sim\mathcal{IG}2\left( \underline{s}_{s_A}, \underline{\nu}_{s_A} \right),
\end{align}
where $\underline{\mathbf{m}}_{A.n} = \begin{bmatrix}\mathbf{e}_{n.N}' & \mathbf{0}_{N(p-1)+d}' \end{bmatrix}$,  $\underline{\boldsymbol{\Omega}}_A$ is a diagonal matrix with vector $\begin{bmatrix}\mathbf{p}^{-2\prime}\otimes\boldsymbol{\imath}_N' & 100\boldsymbol{\imath}_d'\end{bmatrix}'$ on the main diagonal, hence, incorporating the ideas of the Minnesota prior of \cite{Doan1984}. We set $\underline{\nu}_A, \underline{\nu}_{\gamma_A}, \underline{s}_{s_A},$ and $\underline{\nu}_{s_A}$ all equal to 10, which facilitates relatively strong shrinkage that gets updated, nevertheless. Providing sufficient flexibility on this 3-level hierarchical prior distribution occurred essential for a robust shape of the estimated impulse responses.

Each row of the transition probabilities matrix $\mathbf{P}$ of the Markov process follows independently a Dirichlet distribution, $[\mathbf{P}]_{m\cdot} \sim\mathcal{D}irichlet(\boldsymbol{\imath}_M + d_m \mathbf{e}_{m.M})$, as does the initial regime probabilities vector $\boldsymbol{\pi}_0$, $\boldsymbol{\pi}_0 \sim\mathcal{D}irichlet(\boldsymbol{\imath}_M)$.
To assure the prior expected regime duration of 12 months, we set $d_m = 11$ for models with $M=2$, and show the robustness of our results to alternative assumptions regarding the expected regime duration of 1  and 42 months.

The prior distribution for the regime-specific volatility of the log-volatility process, $\omega_n(s_t)$, follows a conditional normal distribution, $	\omega_n(m) \mid \sigma^2_{\omega.n} \sim\mathcal{N}\left( 0, \sigma^2_{\omega.n}\right)$, with a gamma distribution on the shrinkage level, $\sigma^2_{\omega.n} \sim\mathcal{G}(0.5, 1)$. The prior distribution for the autoregressive parameter of the log-volatility process is uniform over the stationarity region, $\rho_n \sim\mathcal{U}(-1,1)$.  
Notably, the marginal prior for $\omega_n(m)$ combines extreme prior probability mass around the restriction for homoskedasticity $\omega_n(m)=0$ and fat tails. The latter enables efficient extraction of the heteroskedasticity signal from data assuring that a shock is heteroskedastic within a regime only if the data favour this outcome.  Our results remain unchanged subject to variation of the prior scale of $\sigma^2_{\omega.n}$ ranging from values 0.1 to 2.

This prior setup following \cite{LSUW2022} implies that the structural shocks' conditional variance, $\sigma^2_{n.t}$, have a log-normal-product marginal prior implied by the normal prior distributions for the latent volatility, $h_{n.t}$, and for the volatility of the log-volatility parameters, $\omega_{n}(m)$. The log-normal-product marginal prior has a pole at value one, corresponding to homoskedasticity, fat tails, and a limit at value zero when the conditional variance goes to zero from the right. Importantly, the extreme concentration of the prior mass at one and strong shrinkage toward this value standardises the structural shocks' conditional variances. This assures appropriate scaling in the estimation of the structural matrix. In our empirical model, the standardisation holds true as the estimated variances fluctuate closely around one.

To draw from the posterior distributions we use a Gibbs sampler. We provide an \textbf{R} package \textbf{bsvarTVPs} facilitating the reproduction of our results as well as new applications.\footnote{Access the \textbf{bsvarTVPs} package at \url{https://github.com/donotdespair/bsvarTVPs}}

\section{Empirical evidence on TVI} \label{sec:resultsTVI}

\noindent In this section, we present the empirical evidence on time-variation in US monetary policy shock identification. We report the posterior probabilities for the four policy rules in each of the regimes as well as the estimated regime probabilities disclosing the timing of the rules' occurrence. We document the remarkable robustness of these results. Our insights summarise the conclusions from over three hundred models we estimated, including over one hundred twenty versions of our benchmark model.

We estimate our models with two Markov states ($M=2$) as data do not support three or four states. The additional states occurred at a negligible number of dates splitting the second regime into more regimes with nonsignificant variations in the values of the regime-specific parameters, which makes them overfitting. Additionally, the very low number of observations within these additional regimes undermines identification via heteroskedasticity and accurate estimation. 

\subsection{TVI probabilities}\label{ssec:TVIprob}

\noindent We find that data support TVI of US monetary policy shocks.  In the two estimated MS regimes, data favor two different generalized Taylor rule identification schemes with remarkably sharp posterior probabilities of the TVI indicators reported in Figure~\ref{fig:TVI}.

\begin{figure}[h]
	\caption{Posterior probabilities of the TVI components, the monetary policy reaction functions, in both regimes} \label{fig:TVI}
	\includegraphics[trim={0.0cm 1.0cm 0.0cm 2.0cm},width=1\textwidth]{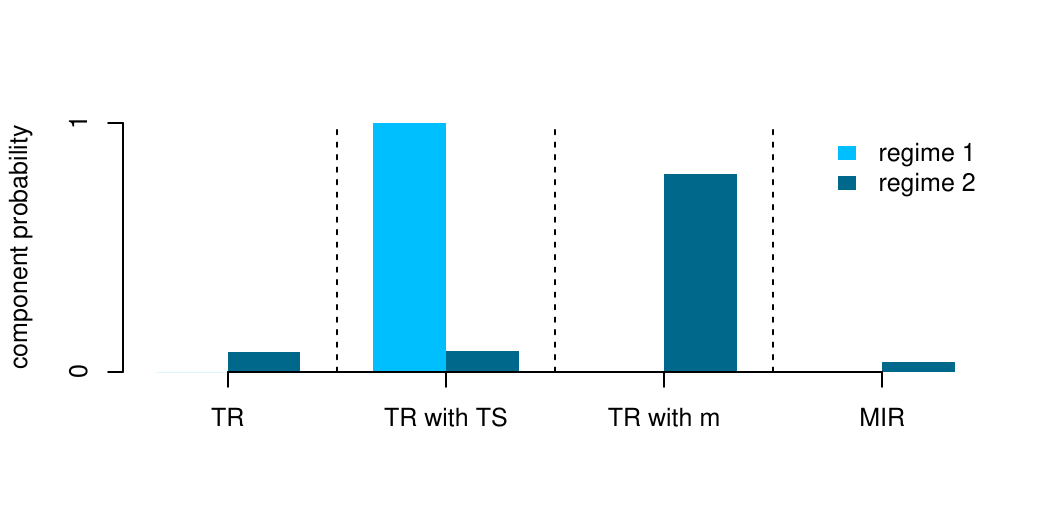} \\ 
	{\footnotesize Note: Figure shows the posterior probabilities of $\kappa_n(m)$ for four identification pattern in the monetary policy reaction function for regime one (light color) and regime two (dark color). US monetary policy shocks are identified via exclusion restrictions imposing on the contemporaneous relations either an empirical Taylor rule setting (``TR''), a TR augmented with term spread (``TR with TS''), a TR augmented with money (``TR with m''), or a money-interest rate rule (``MIR''). }
\end{figure}

In the first regime indicated by the light color in Figure~\ref{fig:TVI}, we find strong data support for including the term spread in the monetary policy reaction function. The posterior probability of the TVI component indicator for the restriction pattern imposing on the contemporaneous relations the TR augmented by term spreads (\textit{TR with TS}) is numerically equal to one. Hence, monetary policy reacts simultaneously to changes in the term spread signaling adjusted expectations on future economic developments.

Several papers argue that term spreads are a target of unconventional monetary policy as long term interest rates are the main transmission channel of the large asset purchase programs \citep{Baumeister2013,Liu2017,Feldkircher2016,Tillmann2020}. The monetary authority can target term spreads via the portfolio balance channel; increased asset purchases decrease the supply of long-term securities and lower long-term yields. \cite{DAmico2013}, among others, provide evidence that the large asset purchase programs during the zero lower bound period affect long-term yields. Our results extend those findings by showing that the monetary policy reaction function of the Fed accounts for the information revealed by term spreads even before the start of unconventional monetary policy since the first regime occurs mainly before 2008. This suggests that the interest rates are established with consideration to the whole bond yield curve in the spirit of \cite{diebold2006}. 

In the second regime visualized by the dark color in Figure~\ref{fig:TVI}, the monetary authority responds simultaneously to changes in the monetary aggregate. The posterior value of the TVI component indicator for the exclusion restriction pattern based on a TR augmented by money  (\textit{TR with m}) is 0.8. The posterior probability of the remaining components is very low, with \textit{TR} and \textit{TR with TS} both having 0.08 posterior probability and the money-interest rate rule 0.04.  

Our findings are in line with \cite{BelongiaIreland2015} supporting that after 2008 monetary policy actions lead to changes in the monetary aggregates. While the Fed officially targeted money supply before 1982, \cite{BelongiaIreland2015} argue that unconventional monetary policy can be seen as attempts to increase money growth. Similar to the strong support of data for money in the monetary policy rule in the second regime, the authors show that excluding a monetary measure from the interest rate rule is rejected by the data in the sample to 2007. Similarly, \cite{LutkepohlWozniak2017} find empirical evidence for monetary policy rules containing money in models identified via heteroskedasticity for data up to 2013. Neither of these papers considers rules including the term spread or time variation of the structural matrix.

The posterior means of the contemporaneous coefficients in the monetary policy reaction function, the third row of the structural matrix, with the bounds of the 90\% highest density interval in parentheses,
\begin{align*}
	\text{Regime 1:}& \quad \underset{(-0.04;\ 0.23)}{0.10} y_t \underset{(-0.11;\ -0.04)}{-0.07} \pi_t   \underset{(3.23;\ 3.84)}{+3.53}  R_t  \underset{(3.71;\ 4.27)}{+3.99} TS_t \\
	\text{Regime 2:}& \quad \underset{(-0.66;\ -0.31)}{-0.48} y_t \underset{(-0.09;\ 0.03)}{-0.03} \pi_t   \underset{(8.91;\ 16.36)}{+12.60}  R_t  \underset{(-1.75;\ -0.41)}{-1.13} m_t
\end{align*} 
show notable differences across the regimes. In the first regime, interest rates and term spreads are clearly contemporaneously related. The posterior mean of the contemporaneous coefficient on term spreads in the monetary reaction function is positive, with strong evidence that it is different from zero. Interest rates react also to inflation, whereas the coefficient on output is not different from zero. In the second regime, money aggregate plays the dominant role in the policy rule for interest rates, as evidenced by a significant and relatively highest in absolute terms coefficient on this variable. This regime also observes an increase in the importance of output relative to inflation in the monetary policy reaction function. The coefficient on the latter is insignificant. This fact, together with the timing of the second regime, is in line with the shift in the emphasis in the monetary policy after 2000 described by \cite{ba2016}.

\subsection{Robustness of TVI}

\noindent Selecting two different generalized TR identification patterns over time is remarkably robust with respect to extending or limiting the set of possible exclusion restrictions in the search. First, broadening the set of exclusion restrictions leads to the same posterior probabilities of the TVI mechanism in the monetary policy reaction function -- shown in Figure~\ref{fig:TVIrob}(a) for a model which includes a fifth option setting exclusion restrictions according to a TR with stock prices.

\begin{figure}[h]
	\caption{Posterior probabilities of the TVI mechanism in the monetary policy reaction function for extended or restricted TVI specifications} \label{fig:TVIrob}
\subfloat[extended by TR with sp]{\includegraphics[trim={0.0cm 0.5cm 0.0cm 2.0cm},width=0.5\textwidth]{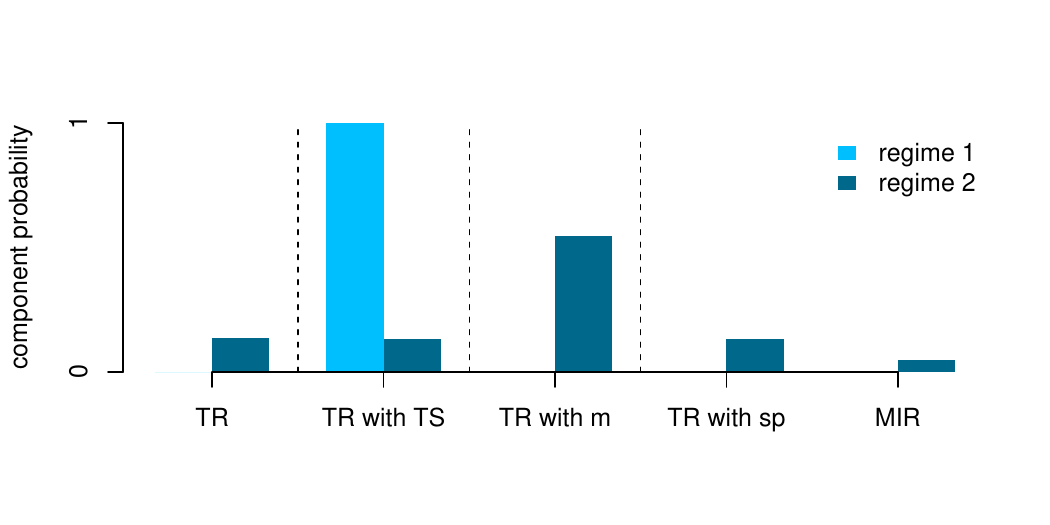}}
 \subfloat[TR and TR with TS]{\includegraphics[trim={0.0cm 0.5cm 0.0cm 2.0cm},width=0.5\textwidth]{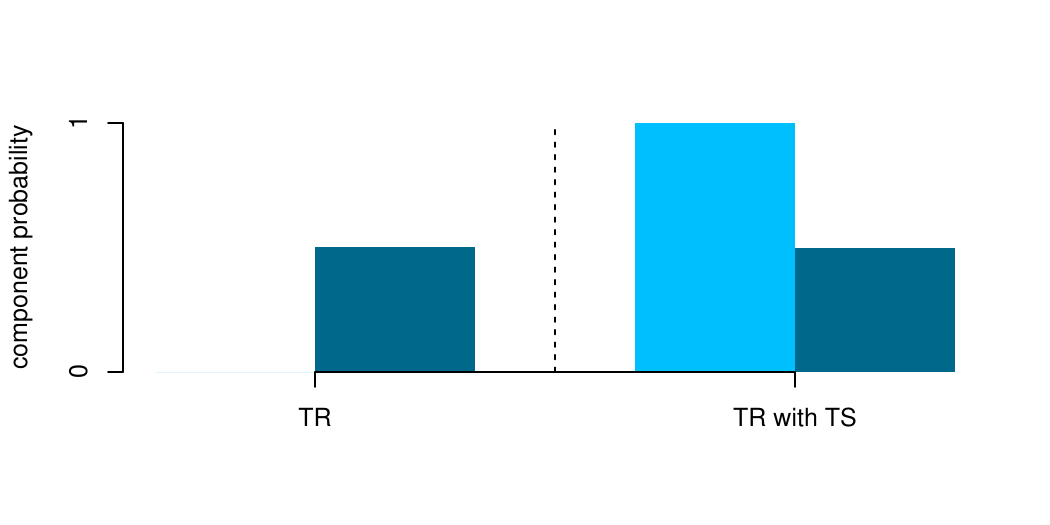}}\\
\subfloat[TR and TR with m]{	\includegraphics[trim={0.0cm 0.5cm 0.0cm 2.0cm},width=0.5\textwidth]{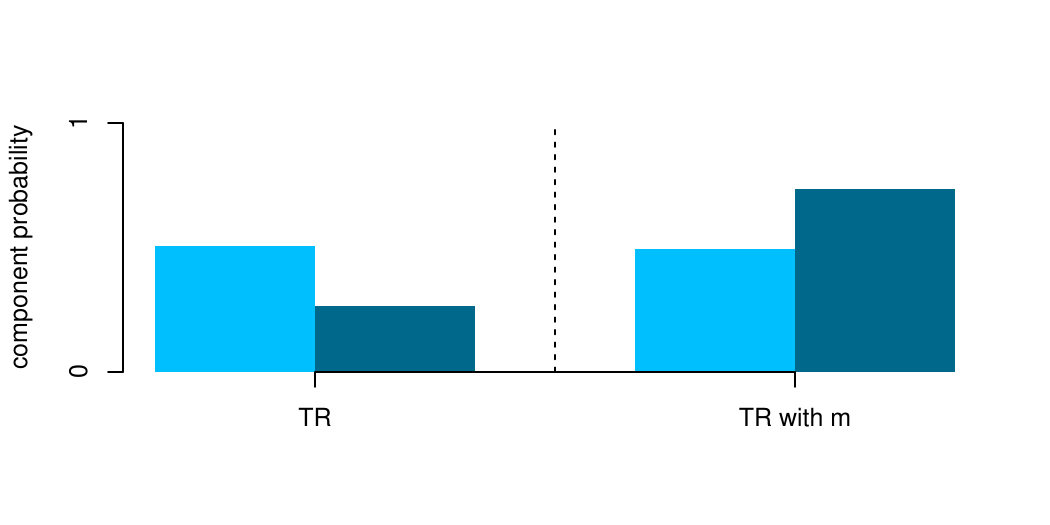}} 
 \subfloat[TR and unrestricted]{	\includegraphics[trim={0.0cm 0.5cm 0.0cm 2.0cm},width=0.5\textwidth]{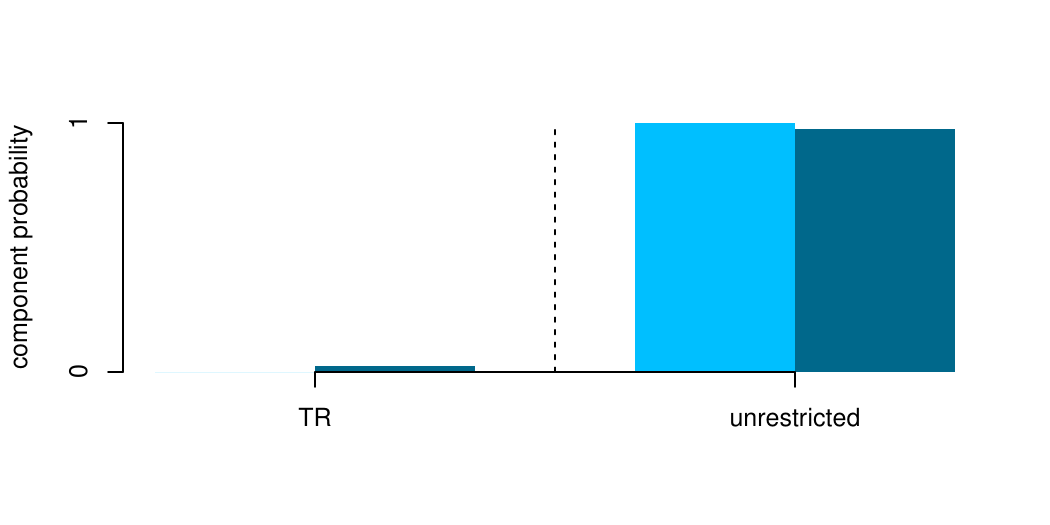}} \\
	{\footnotesize Note: Figure shows the posterior probabilities of $\kappa_n(m)$ for the identification pattern in the monetary policy reaction function for regime one (light color) and regime two (dark color) for three different models. Model (a) allows for TR, TR with TS, TR with m, TR with sp, and MIR. Model (b) allows for TR and TR with TS, model (c) for TR and TR with m, and model (d) for TR and unrestricted.}
\end{figure}

Second, limiting the set of possible exclusion restrictions to the TR and one generalized alternative supports the role of the term spread and money in the respective regimes. Allowing for either the \textit{TR} or \textit{TR with TS} as exclusion restriction patterns data support the generalized TR in the first regime with 100\% posterior probability as shown in Figure~\ref{fig:TVIrob}(b). Absent the relevant alternative, the search in the second regime is uninformative placing equal probability on both specifications. Imposing exclusion restrictions either in line with the \textit{TR} or \textit{TR with m}, the posterior probability of the \textit{TR with m} is around 80\% in the second regime, exactly matching the result from our benchmark model, while it is 50\% in the first regime in Figure~\ref{fig:TVIrob}(c). 

When the TR is compared with an unrestricted third row of the structural matrix, the latter is strongly preferred in both regimes as shown in Figure~\ref{fig:TVIrob}(d). This shows that a rule containing the relevant alternatives such at \textit{TR with TS} or \textit{TR with m} is also preferred over the simple \textit{TR}. Our attempts of estimating a model with an unrestricted policy rule resulted in less persistent regimes and lack of within-regime identification via heteroskedasticity. Overall, this exercise documents that whenever a policy rule preferred by the data or a rule containing it is available, the result is the same as in our benchmark model, and at the absence of it any policy rule is evenly likely.


Moreover, our analysis of an extensive set of alternative specifications changing the model setup reinforces our conclusions regarding the TVI probabilities. These alternative model specifications included those with different lag lengths from one to 12 or changes in the variables included. We used personal consumption expenditures excluding food and energy or an interpolated GDP deflator instead of CPI, interpolated GDP instead of the industrial production index, shadow rates instead of federal funds rates, monetary base as an alternative money aggregate, and added credit spreads or inflation expectations to the model. Furthermore, we changed prior set-ups by changing hyper-parameters on the priors set on the contemporaneous, autoregressive coefficients, SV or the Markov-process with no changes in the TVI probabilities.

\subsection{Interpretation of the two Markov-switching regimes}

\noindent Having established that data clearly support regime-specific identification of US monetary policy shock, we now take a closer look at the two estimated MS regimes.

Figure~\ref{fig:regimeprob} shows the estimated regime probabilities for the second regime (a gray shaded area), together with highlighted periods when the probability for the second regime is larger than 0.8 (dark-colored thick line), and specific events (black vertical lines topped by letters). Regime one is predominant in the first part of the sample until 2000. It is present in periods of rather normal economic developments (aka non-crisis times). After 2000 the second regime occurs more frequent and the periods are characterized by a mixture between the two regimes. The second regime clearly prevails during crisis periods, in particular the financial and COVID-19 ones. In the aftermaths of these events, regime 2 gains persistence. Additionally, it is present at earlier extra-ordinary times such as the peak of the 1981--1982 crisis, the burst of the dot-com bubble in March 2000, or 9/11. Beside its characterization as a crisis regime, it also predominately covers periods of unconventional monetary policy actions such as Quantitative Easing 1, Operation Twist, and several subsequent announcements of the Fed, starting in August 2011, stating that they would keep the federal funds rate at its effective zero-lower bound for a substantial period.

\begin{figure}[t!]
	\caption{Regime probabilities of the Markov process} \label{fig:regimeprob}
		\includegraphics[trim={1.5cm 0.5cm 1.0cm 2.0cm},width=1\textwidth]{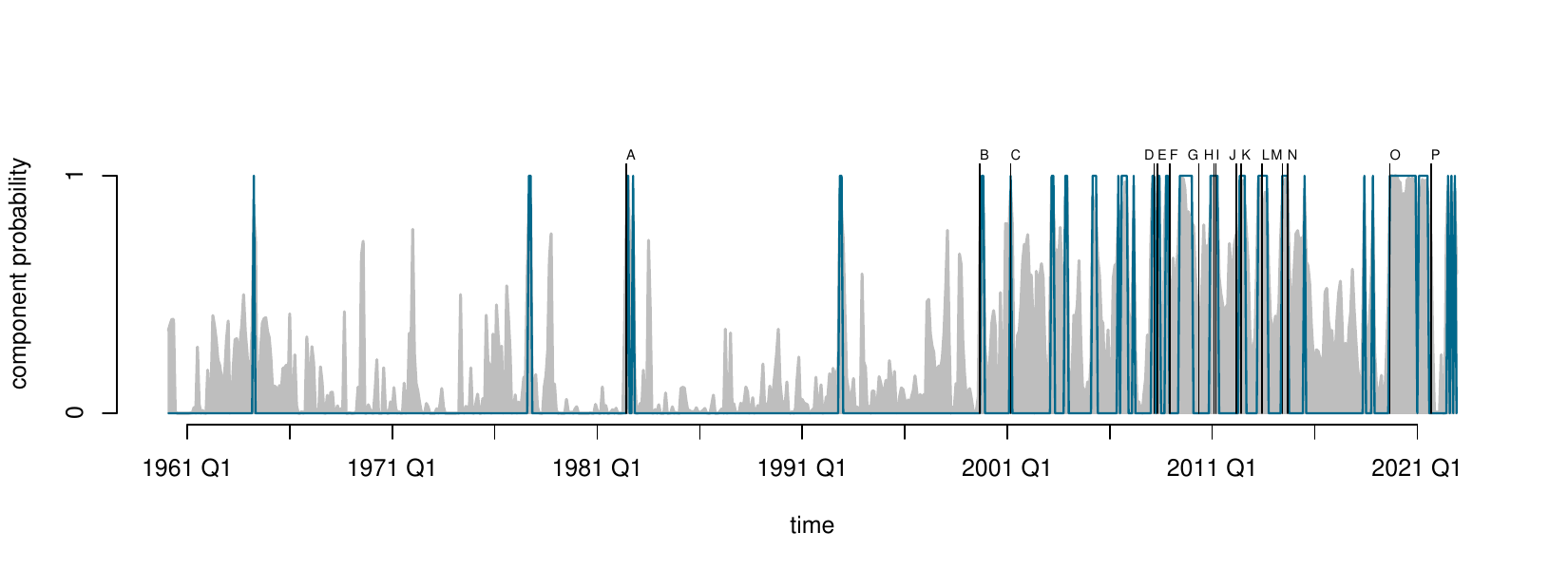} \\
	{\footnotesize Note: Figure shows the estimated regime probabilities of the second regime  (shaded areas) and the occurrences of the second regime with probability larger than 0.8 (dark color). Probabilities for the first regime are one minus the probabilities of the second regime.\\ \textbf{List of events:}\\[1ex] \footnotesize
	\begin{tabular}{p{0.4cm}p{1.5cm}p{13cm}} \toprule
		A &  Dec 1982& peak of 1981-1982 recession\\
		B	& Mar 2000 & dot-com bubble burst \\
		C	& Sep 2001  & September 11 attacks\\
		D & Sep 2008 & bankruptcy of Lehman Brothers on September 15, 2008 \\
		E & Nov 2008  & FOMC announces Quantitative Easing 1\\
		F & Jun 2009 & NBER declares end of the US great recession\\
		G	&		Nov 2010 & FOMC announces Quantitative Easing 2 \\
		H	&		Aug 2011 & FOMC announces to keep the federal funds rate  at its effective zero lower bound ``at least through mid-2013''$^*$\\ 
		I & Sep 2011  & FOMC announces Operation Twist\\
		J	&			Sept 2012 & FOMC announces Quantitative Easing 3 \\
		K	&		Dec 2012 & FOMC announces to purchase 45B\$ of longer-term Treasuries per month for the future, and to keep the federal funds rate at its effective zero lower bound\\\
		L	&	Dec 2013 & FOMC announces to start lowering the purchases of longer-term Treasuries and mortgage-backed securities (to \$40B and \$35B per month, respectively)\\
		M	&		Dec 2014 & FOMC announces  ``it can be patient in beginning to normalize the stance of monetary policy''\\
		N &		Mar 2015 & FOMC announces ``an increase in the target range for the federal funds rate remains unlikely at the April FOMC meeting''\\
		O & Mar 2020& start of COVID-19 pandemic. First unscheduled FOMC meeting on March 3 in response to the outbreak of the pandemic announcing to ``lower the target range for the federal funds rate by 1/2 percentage point'' and to ``use its tools and act as appropriate to support the economy.'' On March 23, FOMC statement to ``continue to purchase Treasury securities and agency mortgage-backed securities in the amounts needed to support smooth market functioning''.\\
		P & Mar 2022& First indications of changes in monetary policy with FOMC annoucing that one member voted against keeping the federal funds rate at the same rate. In May and June the federal funds rate was increased with the largest hikes since May 2000 and 1994, respectively. \\
		\bottomrule
\end{tabular}\\ $^*$ all quotes from \url{https://www.federalreserve.gov/newsevents/pressreleases.htm}.}
\end{figure}

As revealed by TVI, the Fed acts on the federal funds rate simultaneously with term spread movements in the first regime and money in the second. Moreover, regime-specific data and structural shocks' sample moments reported in Table~\ref{tab:regimeparam} attach an crisis versus non-crisis interpretation to the regimes. The regime-specific means and standard deviations of the variables display patterns characteristic to a lower level of economic activity and higher volatility. This is particularly well-evidenced for the industrial production growth rate and inflation that have considerably lower means and higher standard deviations in the second regime. The federal funds rate is on average lower and the term spread higher in the second regime while they both exhibit lower volatility then. Money is twice as volatile in the second regime. Remarkably, the differences for $sp_t$ are small indicating that stock market developments explain rather little in the differences across the regimes.

\begin{table}[t]
	\caption{Regime-specific sample moments of data and structural shocks} \label{tab:regimeparam}
	\centering
	\begin{tabular}{lccccccccccc}
		\toprule
& \multicolumn{5}{c}{data}  &&  \multicolumn{5}{c}{structural shocks}\\
\cline{2-6} \cline{8-12}
& \multicolumn{2}{c}{Regime 1} && \multicolumn{2}{c}{Regime 2} && \multicolumn{2}{c}{Regime 1} && \multicolumn{2}{c}{Regime 2}\\   
& mean & sd &&  mean & sd && mean & sd && mean & sd  \\ \cline{2-3} \cline{5-6} \cline{8-9} \cline{11-12}
$y_t$ & 2.72& [8.40] && 1.46& [17.36] && 0.00& [0.97] && -0.11& [1.65] \\ 
  $\pi_t$ & 4.11& [3.61] && 2.65& [4.04] && 0.00& [1.01] && -0.04& [1.13] \\ 
  $R_t$ & 5.67& [3.66] && 2.54& [2.74] && 0.03& [1.05] && -0.14& [1.23] \\ 
  $TS_t$ & 0.87& [1.69] && 1.48& [1.28] && -0.02& [1.61] && -0.07& [0.76] \\ 
  $m_t$ & 6.45& [3.84] && 7.41& [8.53] && -0.02& [0.90] && 0.08& [1.23] \\ 
  $sp_t$ & 7.24& [51.14] && 6.15& [54.58] && 0.01& [0.99] && -0.06& [1.01] \\ 
	\bottomrule
	\end{tabular}
	\begin{flushleft}
		\footnotesize Note: Table reports sample means and standard deviations for variables and estimated structural shocks in both regimes. The regime-specific moments are given for the series in first differences for $y_t$, $m_t$, and $sp_t$.
	\end{flushleft}
\end{table}

The sample moments characterising the structural shocks within the regimes reported in Table~\ref{tab:regimeparam} confirm our interpretation of the second regime as the crisis one. As long as the regime-specific averages of the shocks are equal to zero in the first regime, the signs of all the shocks in the second one are negative, except for the money supply shock. The standard deviations are greater than those in the first regime, except for the term spread shock. The negative averages for industrial production, inflation, and stock price shocks indicate strong recessionary pressures from the economy in the second regime. Particularly indicative is the high in absolute terms and a negative average of the monetary policy shock reaching -14 basis points and that of the term spreads, equal to -7 basis points. This indicates that the second regime captures monetary policy that is characterized by cutting the interest rates and narrowing the bond yield curve supplemented by expanding the monetary basis. These findings confirm our labelling of the second regime as more erratic crisis periods that trigger unconventional monetary policy actions.

We find that the regime allocations are robust to various changes in the model specification. We illustrate the role of important factors for the regime allocation by altering the measure of the monetary policy indicator and inflation, changing the prior expected duration of the Markov process, and discarding TVI. The estimated regime probabilities, shown in Figure~\ref{fig:regimeprob_rob}, are very similar for all alternatives: for models including shadow rates of \cite{Wu2016} as an alternative measures of short term interest rates (labeled as ``with shadow rates''), using GDP deflator as measure of inflation (``with GDP deflator''), setting the expected regime duration of the Markov process to 42 months (``prior MS''),  with MS but without TVI where we impose on the structural matrix restrictions according to either \textit{TR}, \textit{TR with TS} or \textit{TR with m} (labeled by ``no TVI with TR'', ``no TVI with TR with TS'', ``no TVI with TR with m'', respectively).

\begin{figure}[t]
	\caption{Regime probabilities of the Markov process across alternative model specifications} \label{fig:regimeprob_rob}
		\includegraphics[trim={1.5cm 0.5cm 1.0cm 2.0cm},width=1\textwidth]{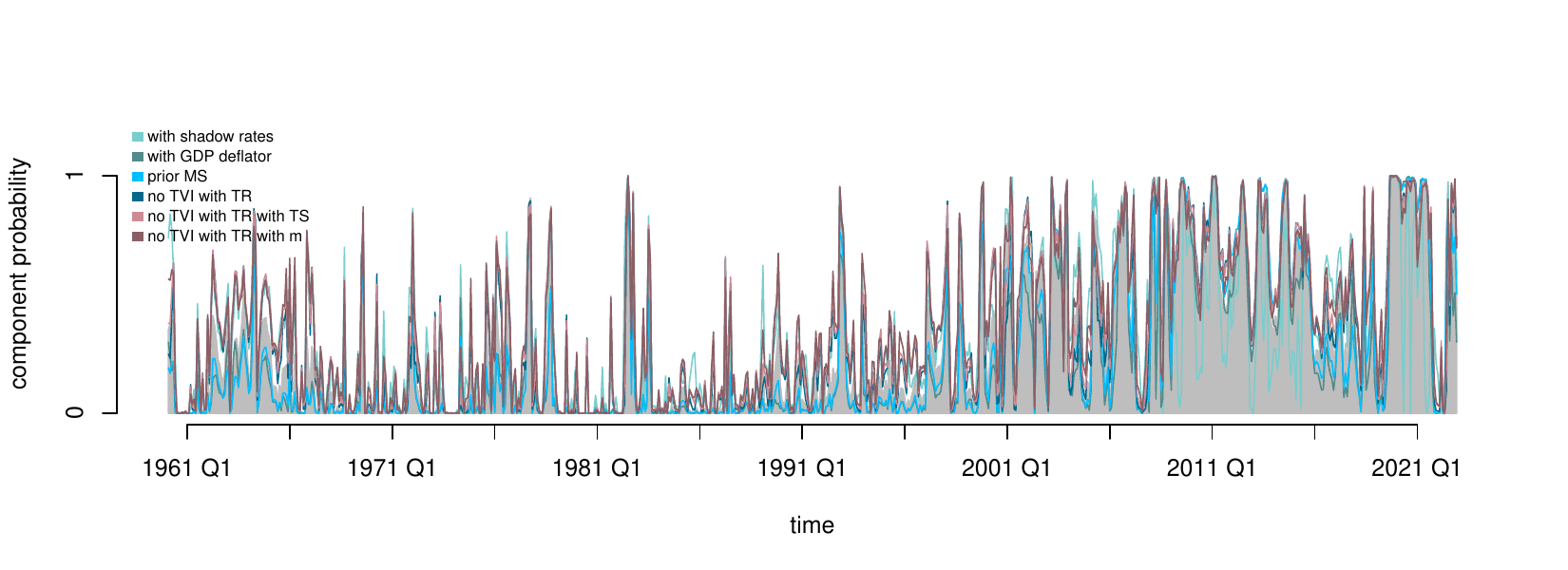} \\
	{\footnotesize Note: Figure shows the estimated regime probabilities of the second regime for several model specifications (solid lines) together with those of the base line model (shaded area). Probabilities for the first regime are one minus the probabilities of the second regime.  }
\end{figure}

\subsection{Within-regime identification via heteroskedasticity}

\noindent Our search for regime-specific identification patterns requires that the exclusion restrictions on the monetary policy rule over-identify it and, thus, can be tested using statistical methods. We achieve this by identification of the monetary policy shock via heteroskedasticity. Figure~\ref{fig:omega} plots histograms of the monetary policy shock volatility of the log-volatility parameter $\omega_3(m)$ for both regimes, $m\in\{1,2\}$. Verifying whether these parameters differ from zero implies identifying the underlying shocks. We find clear evidence for heteroskedasticity within both regimes for the monetary policy reaction function visualized by the bi-modality of the posterior distribution. The mass of the posterior distributions is away from zero, which is more pronounced in the second regime. Therefore, the monetary policy shock is identified via heteroskedasticity and our selection of identifying restrictions is valid in both regimes.

\begin{figure}[t]
	\caption{Marginal posterior and prior densities of the regime-specific volatility of the log-volatility parameter, $\omega_{3}$, for the monetary policy shock} \label{fig:omega}
	\includegraphics[trim={0.0cm 0.5cm 0.0cm 2.0cm},width=1\textwidth]{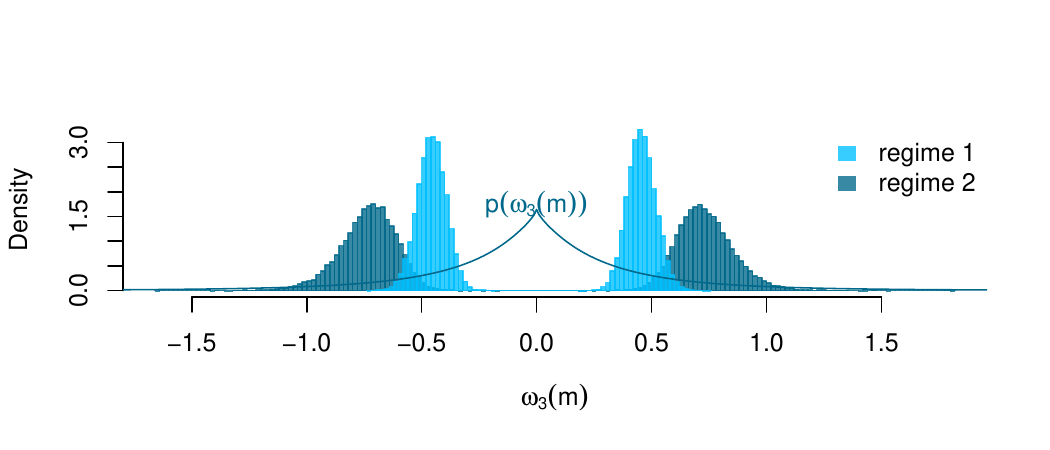} \\
	{\footnotesize Note: Figure reports histograms of the posteriors for the regime-specific monetary policy shock volatility of the log-volatility parameters $\omega_3(m)$ for regime one (light color) and regime two (dark color). The marginal prior distribution is given by the solid line. The concentration of the posterior probability mass away from zero indicates that the monetary policy shock is identified via heteroskedasticity in both regimes.}
\end{figure}

We then label the third shock as the monetary policy shock based on the exclusion restrictions imposed on the monetary policy reaction function. This is sufficient in the second regime featuring the money-augmented TR as the set of restrictions in the corresponding row of the structural matrix is unique. However, in the first regime, the exclusion restrictions selected for the third row of the structural matrix are the same as those for the fourth. Here, we additionally consider the impulse response functions to label the third structural shock as the monetary policy shock.

We find that the identification of the monetary policy shocks via heteroskedasticity is robust to changes of hyper-parameters in the prior distributions on contemporaneous effects, autoregressive coefficients, SV, or the Markov process. Crucially, changing the scale hyper-parameter of the prior controlling for the volatility of the log-volatility level by order of magnitude (we tried values ranging from 0.1 to 2) only marginally affects the posterior distribution of $\omega_{3}(m)$ in both regimes. Also, using alternative measurements for our variables, such as shadow rates instead of federal funds rates or an interpolated GDP deflator for CPI inflation, has no impact on identification via heteroskedasticity. However, there are some factors impeding identification via heteroskedasticity especially in the second regime. They include, for instance, using interpolated GDP as an economic activity measure or specifying a model with MS in the structural matrix but without TVI. 

\section{TVI Effects of US monetary policy shocks}  \label{sec:resultsIRF}

\noindent In this section we present regime-specific impulse response to a US monetary policy shock. We then focus on the role of the term spread in more detail, analyze time-variation of the effects of monetary policy shocks and a counterfactual scenario in which the monetary policy in the second regime does not emerge.

\subsection{Dynamic responses to US monetary policy shocks}

\noindent We find notable differences in the dynamic responses to the US monetary policy shock across the regimes. They lead to distinguishing interpretations of the monetary policy conducted in the particular sub-periods. Based on its effects in the first regime, monetary policy can be characterized as conventionally inflation-targeting, whereas in the second one, as strongly expansionary with specific macro-financial interactions. Our analysis is based on Figure~\ref{fig:IRF} reporting the median impulse responses to a monetary policy shock decreasing interest rates by 25 basis points in the first (light color) and second (dark color) regime together with the 95\% posterior highest density sets. As in the second regime the monetary policy shocks are on average negative as reported in Table~\ref{tab:regimeparam}, we discuss the effects of a negative shock.

\begin{figure}[t]
	\caption{Impulse response to US monetary policy shocks} \label{fig:IRF}
	\subfloat[response of $y$]{\includegraphics[trim={0.0cm 0.5cm 0.0cm 2.0cm},width=0.33\textwidth]{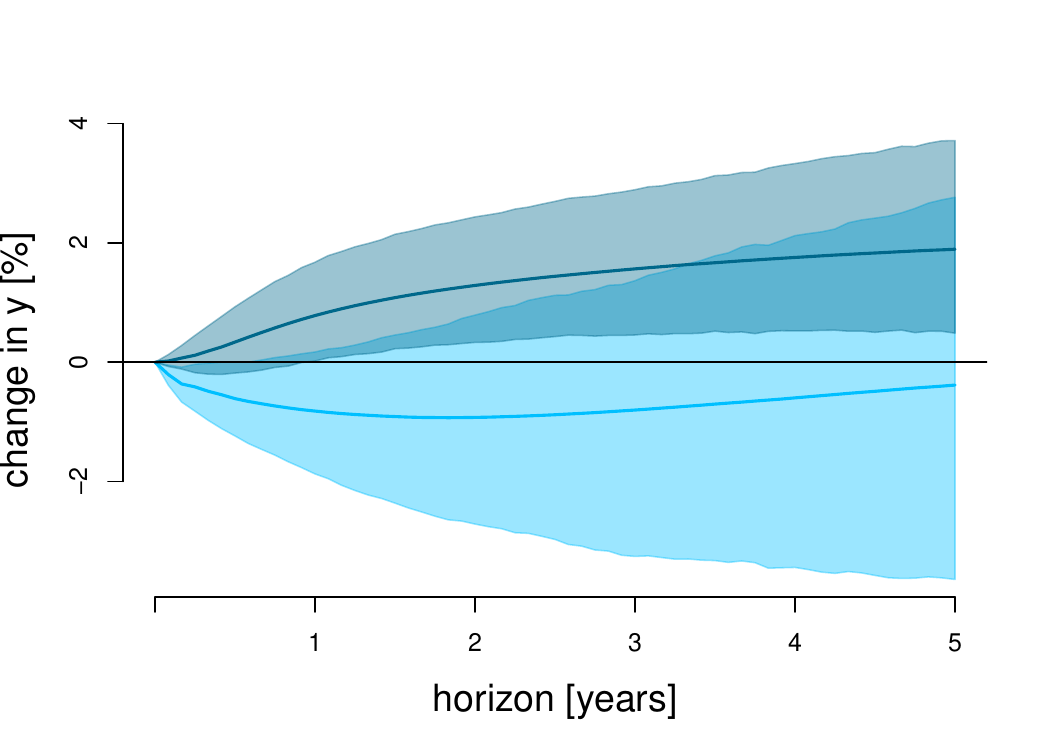}}
	\subfloat[response of $\pi$]{\includegraphics[trim={0.0cm 0.5cm 0.0cm 2.0cm},width=0.33\textwidth]{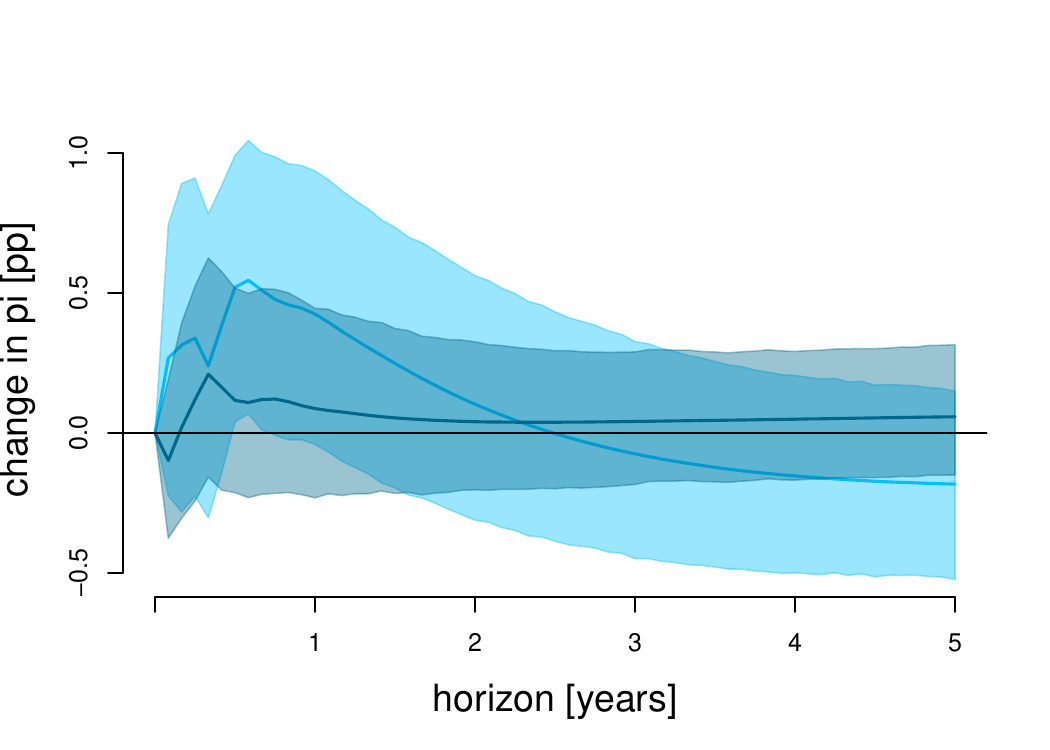}}
	\subfloat[response of $R$]{	\includegraphics[trim={0.0cm 0.5cm 0.0cm 2.0cm},width=0.33\textwidth]{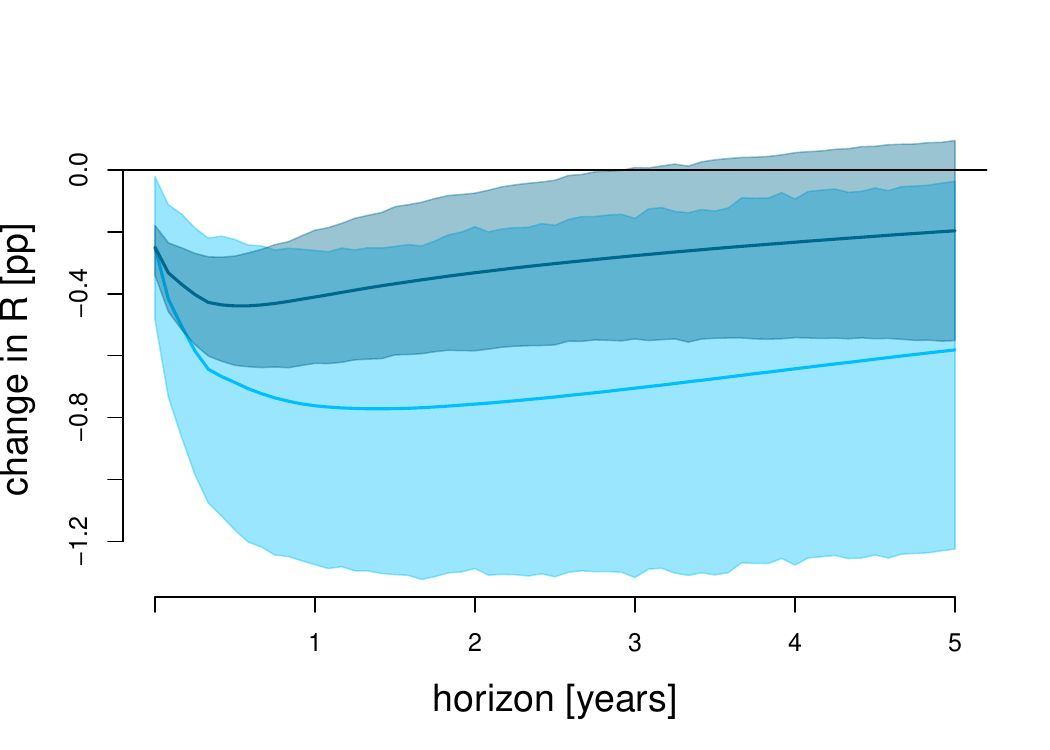}} \\
	\subfloat[response of $TS$]{	\includegraphics[trim={0.0cm 0.5cm 0.0cm 2.0cm},width=0.33\textwidth]{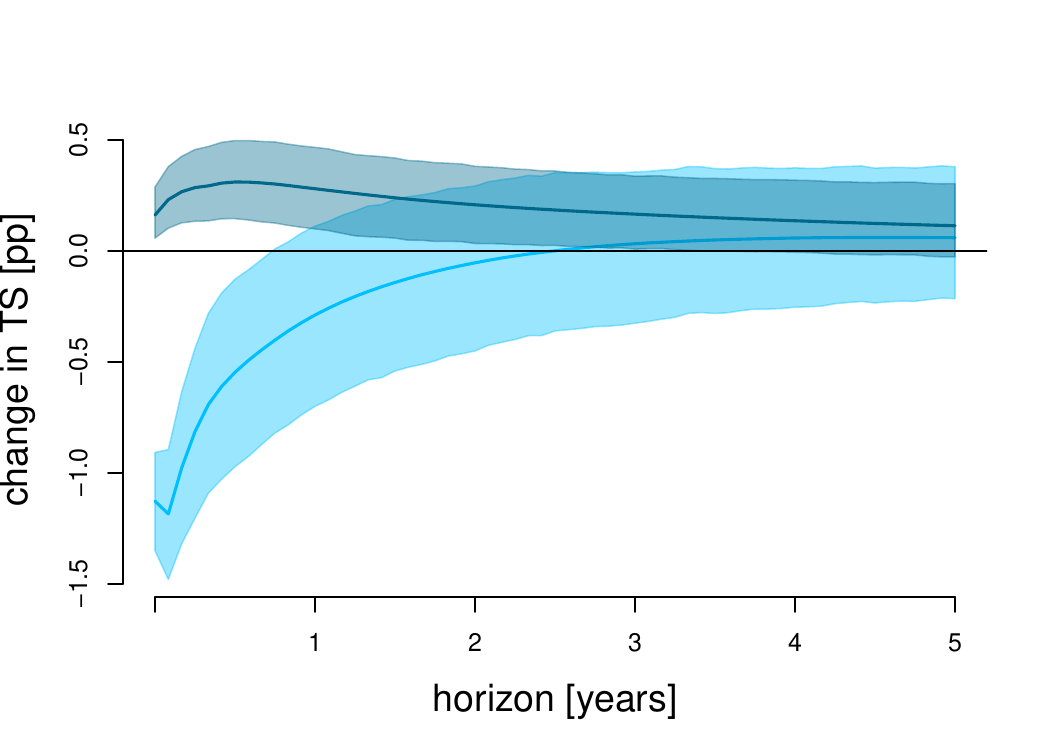}}
	\subfloat[response of $m$]{	\includegraphics[trim={0.0cm 0.5cm 0.0cm 2.0cm},width=0.33\textwidth]{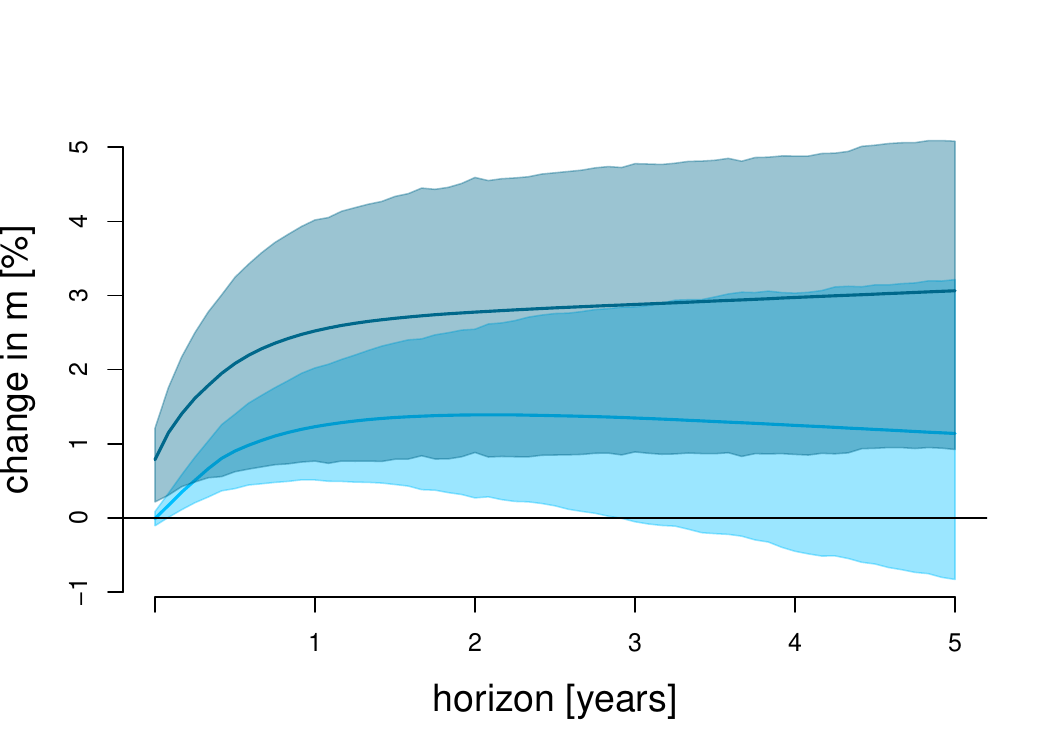}}
	\subfloat[response of $sp$]{	\includegraphics[trim={0.0cm 0.5cm 0.0cm 2.0cm},width=0.33\textwidth]{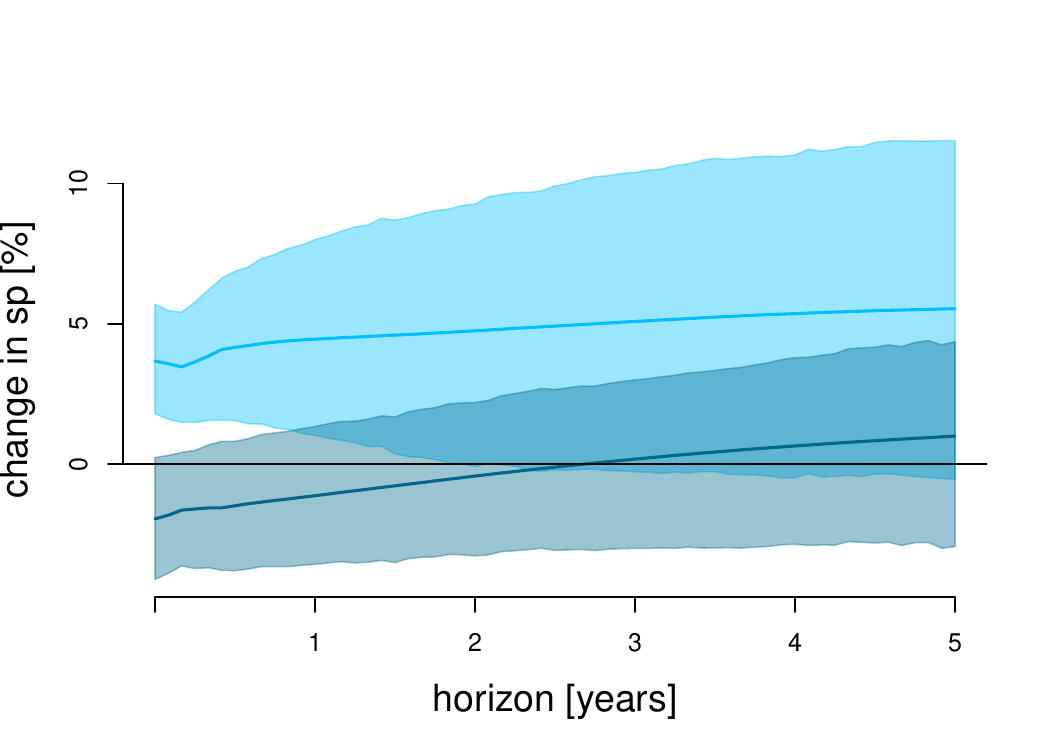}}
	\\	{\footnotesize Note: Figure shows posterior median impulse responses to a expansionary monetary policy shock of 25 basis points (solid lines) for regime 1 (light color) and regime 2 (dark color) together with 95\% posterior highest density sets (shaded areas).}
\end{figure}

In the first regime, an unexpected monetary loosening leads to a persistent level shift in the federal funds rate. The response falls to the lowest point after 15 months at 77 basis points. This shock is effective in increasing inflation within the perspective of the following 19 months and with the highest impact of 0.55 pp after seven months. The effect on output is statistically insignificant. In light of the regime characteristics, we interpret the reaction, which seems to be ad odds with (linear) empirical evidence, as showing that monetary loosing fails to provide economic stimuli in non-crisis periods. It suggests that the inflationary push overrules potential direct effects of interest rate decreases in rather calm and economically stable periods. The system features the liquidity effect with the money supply substantially and permanently increasing. In line with \cite{Ang2011}, the lower interest rates decrease the slope of the bond yield curve for nearly a year, with the trough response narrowing the term spread by 118 basis points after one month. They also lead to a sharp increase in stock prices. 

Moreover, in line with the inflation-targeting goal, our results imply that an unexpected monetary tightening is effective in curbing inflation in the first regime. An increase in the federal funds rate induced by a positive monetary policy shock would result in a statistically significant decrease in inflation reaching up to 0.55 pp after seven months.

In the second regime, a negative monetary policy shock decreases the federal funds rate permanently but slightly more moderately than in the first regime. In this regime, the monetary policy is strongly expansionary as output increases permanently and by nearly 1\% within the first year. At the same time, inflation's response is negligible and only significant at the horizon of four months following the shock reaching 0.2 pp. These phenomena are in line with the Fed's shifting emphasis from inflation to output after 2000, discussed in Section~\ref{ssec:TVIprob}. 

In these times of crises, the liquidity effect is positive and even more pronounced in this regime with the response stronger by 1 pp nearly at all horizons. The slope of the bond yield curve is only slightly affected as the term spread widens by 30 basis points after one month, strengthening the conclusion by \cite{Tillmann2020} on time variation in this effect. The response cools off gradually but remains significant at all horizons. The stock prices contract by nearly 2\% on impact -- the effect not being statistically significant in the second regime. This result closely follows the change over time of the impulse responses reported by \cite{gali2015} that the authors explained by the existence of rational bubbles in the stock market.

These effects are in line with the characterisation of the monetary policy in the US as unconventional in the periods aligning with the second regime. The main objectives of these policies were to stimulate the economy in the aftermaths of the crises and provide liquidity, all of which was implemented without substantially increasing inflation in line with ``missing inflation`` during the 2010s \citep[see][]{bobeica2019}, a period which is covered by the second regime. 

Our results indicate time-variation in the monetary policy shock identification. Time-varying parameters and identification play a crucial role in the interpretability of the shocks. It is not obtained if any of these features are omitted from the model formulation. Nevertheless, our results are robust to many other alterations in the specification. The impulse responses keep their shapes irrespectively of changing the lag order as long as it is greater than 4, replacing the variables by alternative measurements, such as using interpolated real GDP or shadow rates for industrial production or interest rates, respectively,  or altering the prior hyper-parameter values. Finally,  TVI allows us to obtain distinctive interpretations of the monetary policy effects in both MS regimes. We show other aspects of the monetary policy in the second regime, characterising it as unconventional by analysing the effects of the term spread shock.

Finally, the contractionary monetary policy shock effectively lower inflation and, therefore, our findings do not exhibit a price puzzle. We find two crucial conditions under which our results do not feature the price puzzle. First, we include as measure of inflation the price index in log first differences. Note, that the inflation series includes a unit root as other series do while the log CPI price index is integrated of order two. Second, TVI is essential to avoid model misspecification resulting in the puzzle as it allows for an enriched analysis of the time-varying nature of monetary policy effects and enhances the interpretability of impulse responses.

\subsection{Dynamic responses to US term spread shocks}

\noindent To understand the differences of the term spread response to the monetary policy shock across regimes, we relate it to the responses of interest rates, term spread, and money to a term spread shock, shown in Figure~\ref{fig:IRFts}. The spread shock -- an unexpected narrowing of the term spread -- describes an unanticipated change in the financial markets such as shifts in liquidity or risk preferences, or inflation expectations. When the federal funds rate is at its effective zero lower bound, a term spread shock can be interpreted as a result of the Feds bond purchase programs  \citep{Baumeister2013}. 

\begin{figure}[h]
	\caption{Impulse response to US term spread shocks} \label{fig:IRFts}
	\subfloat[response of $R$]{	\includegraphics[trim={0.0cm 0.5cm 0.0cm 2.0cm},width=0.33\textwidth]{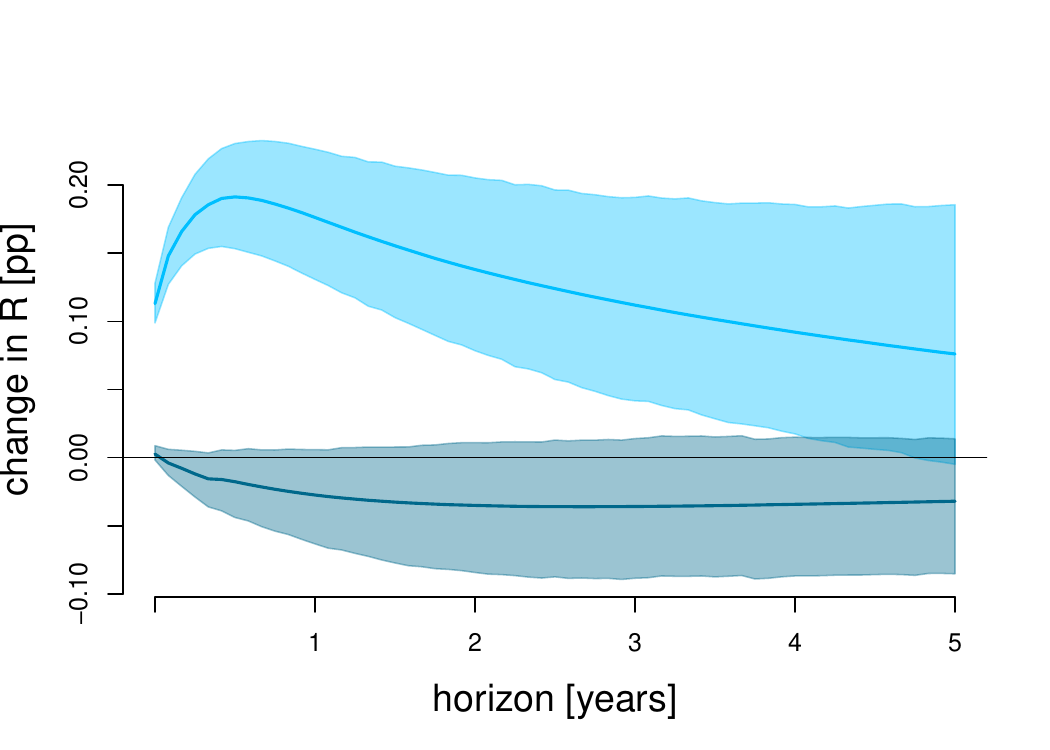}} 
	\subfloat[response of $TS$]{	\includegraphics[trim={0.0cm 0.5cm 0.0cm 2.0cm},width=0.33\textwidth]{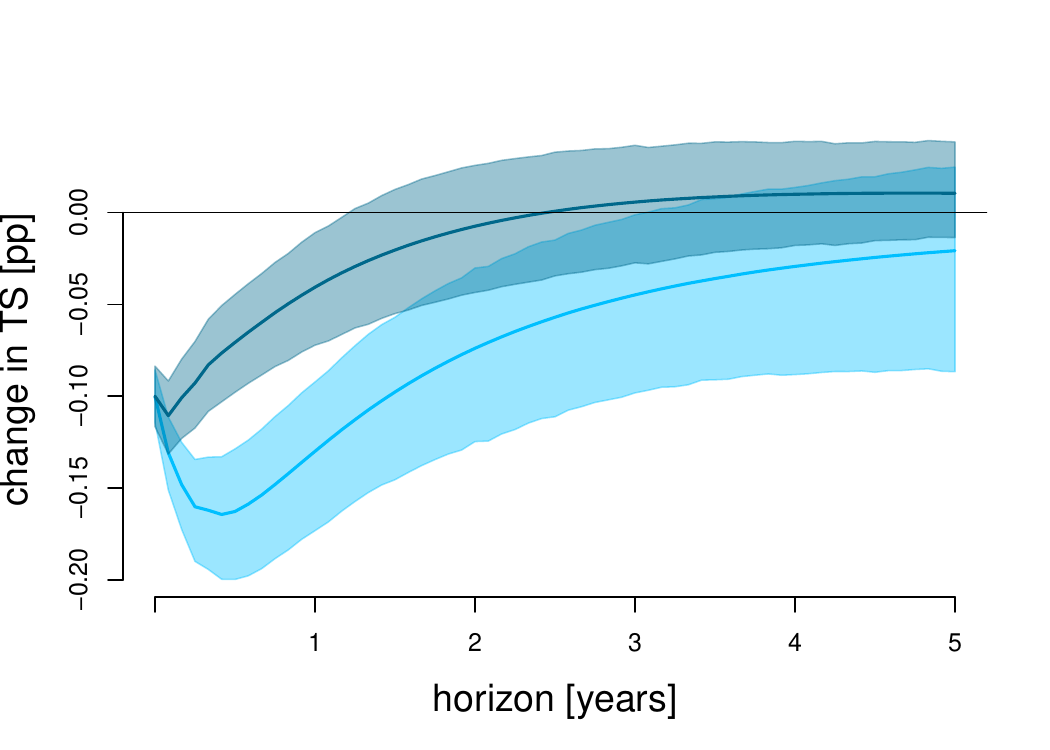}}
	\subfloat[response of $m$]{	\includegraphics[trim={0.0cm 0.5cm 0.0cm 2.0cm},width=0.33\textwidth]{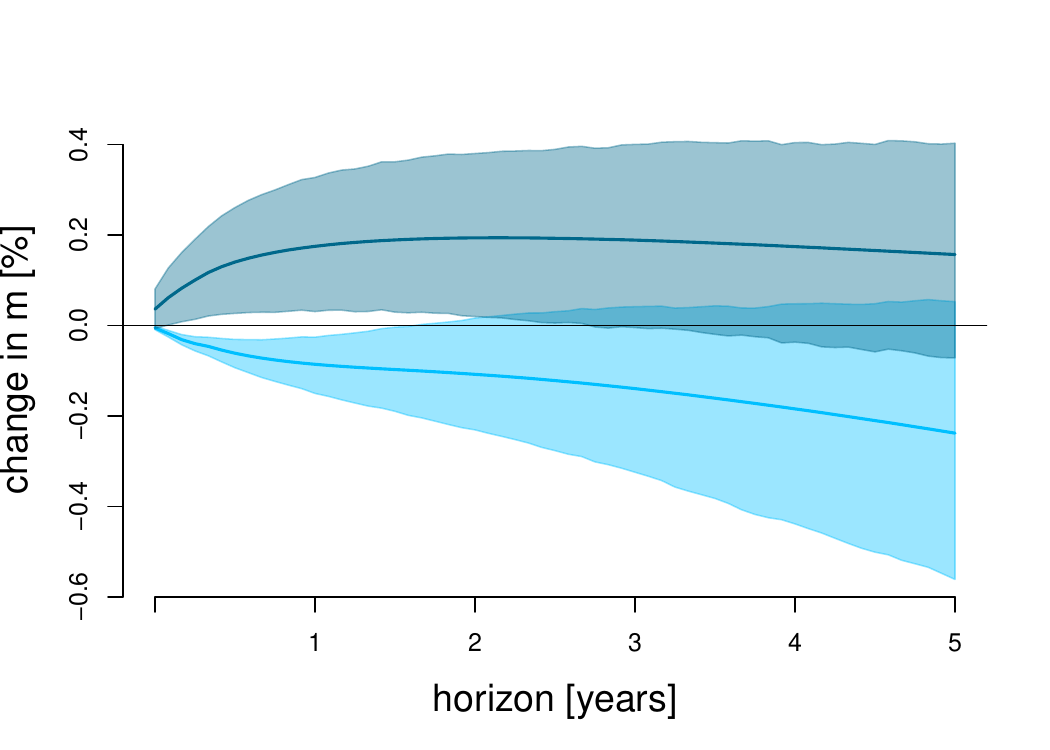}}
	\\	{\footnotesize Note: Figure shows posterior median impulse responses to a negative spread shock of 10 basis points (solid lines) for regime (light color) and regime 2 (dark color) together with 95\% posterior credibility sets (shaded areas).}
\end{figure}

In the first regime, the term spread shock is identified via heteroskedasticity. Since TVI selects Taylor rule with term spread as exclusion restriction pattern for the third row in $\mathbf{B}$, and hence the same restrictions are as set in the fourth row (due to the recursive assumption on the rest of $\mathbf{B}$), the exclusion restrictions on the forth row are not sufficient to label the shocks. We do that by looking at the IRFs. In the second regime, the term spread shock is identified via the exclusion restrictions.

A negative spread shock of 10 basis points lowers terms spreads for at least a year in both regimes. The reaction in the first regime is stronger and lasts longer. We find notable differences across regimes for the response of the federal funds rate and money. 

The short term interest rates increase in the first regime by 0.11 pp on impact and peaks at 0.19 pp within one year. In the second regime, interest rates do not move at any horizon. The differences in the impact response are remarkable, visualized in the histograms of the zero horizon responses in Figure~\ref{fig:histRtoTS}. The posterior estimates are positive and clearly away from zero in the first regime, while in the second the posterior distribution of the estimated impact response is tightly centered around zero with considerable mass at zero. In the second regime, which predominately covers the period of the zero-lower-bound environment,  monetary policy interventions can have a narrowing effect on the term spread while the short term interest rate does not move. That finding provides empirical evidence for the restriction that short term interest rates do not react to a term spread shock for four quarters imposed by \cite{Baumeister2013} and \cite{Feldkircher2016}. However, we found support for the restriction the other authors assume only for the second regime, and we obtained it in the process of estimation using TVI rather than imposing it in a system otherwise identified by sign restrictions. Nevertheless, this result confirms the identification of a \textit{pure} term spread shock flattening the bond yield curve without affecting the interest rate contemporaneously that characterises unconventional monetary policy according to \cite{Baumeister2013}.

\begin{figure}[t]
	\caption{Zero horizon response of interest rates to US term spread shocks} \label{fig:histRtoTS}
 	\includegraphics[trim={0.0cm 0.5cm 0.0cm 2.0cm},width=1\textwidth]{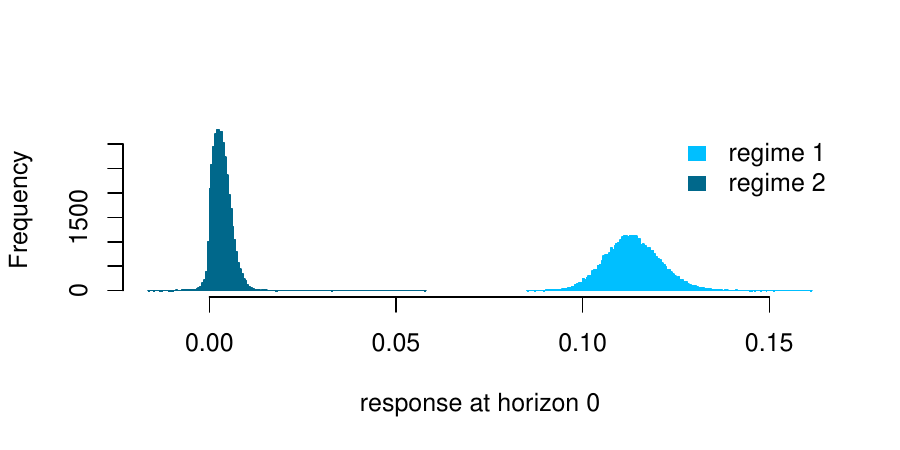}
	\\	{\footnotesize Note: Figure shows posterior distribution of zero horizon response of interest rates to a negative spread shock of 10 basis points for regime (light color) and regime 2 (dark color). For regime 2 only those draws are taken where TR with m was selected.}
\end{figure}

Money increases in reaction to a negative spread shock in the second regime and decreases in the first. The effect lasts in mid-horizons in both regimes. The positive reaction of money in the second regime -- which is predominate during unconventional monetary policy periods -- aligns with the interpretation of the term spread shock as mirroring quantitative easing \citep[see][]{Baumeister2013,Feldkircher2016,Liu2017}.

\subsection{What does monetary policy do?}

\noindent How effective was the Fed in reaching its dual mandate of price stability and economic growth over time? To quantify this, we show the cumulative effects of monetary policy shocks over one year at each point in time for industrial production and inflation in Figure~\ref{fig:EFF}.

\begin{figure}[t]
	\caption{Cumulative effects of the last year's monetary policy shocks over time} \label{fig:EFF}
	\subfloat[$y$]{\includegraphics[trim={0.0cm 1.0cm 0.5cm 2.0cm},width=0.5\textwidth]{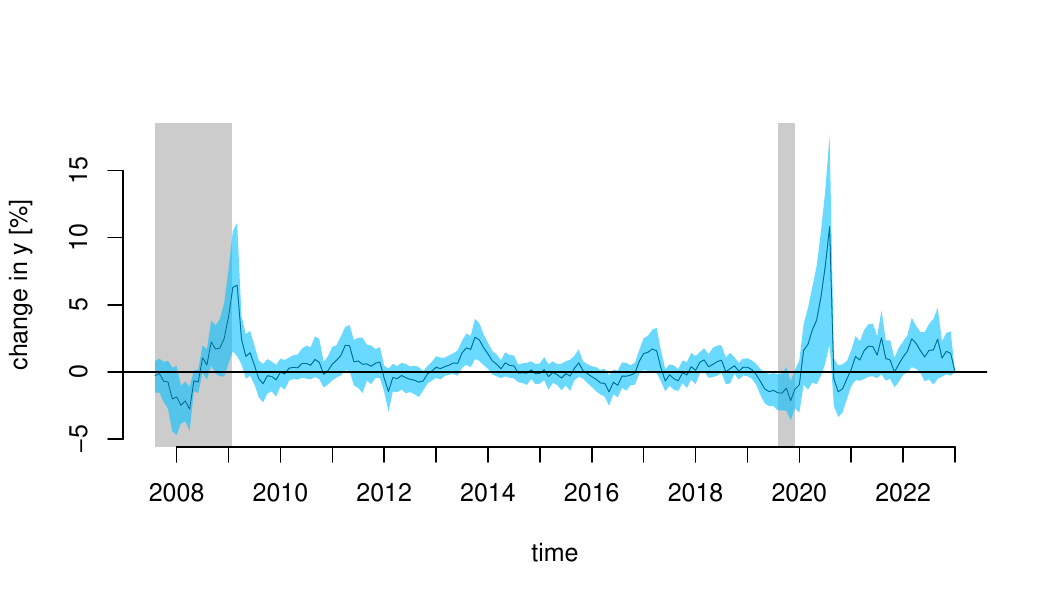}}
	\subfloat[$\pi$]{\includegraphics[trim={0.0cm 1.0cm 0.5cm 2.0cm},width=0.5\textwidth]{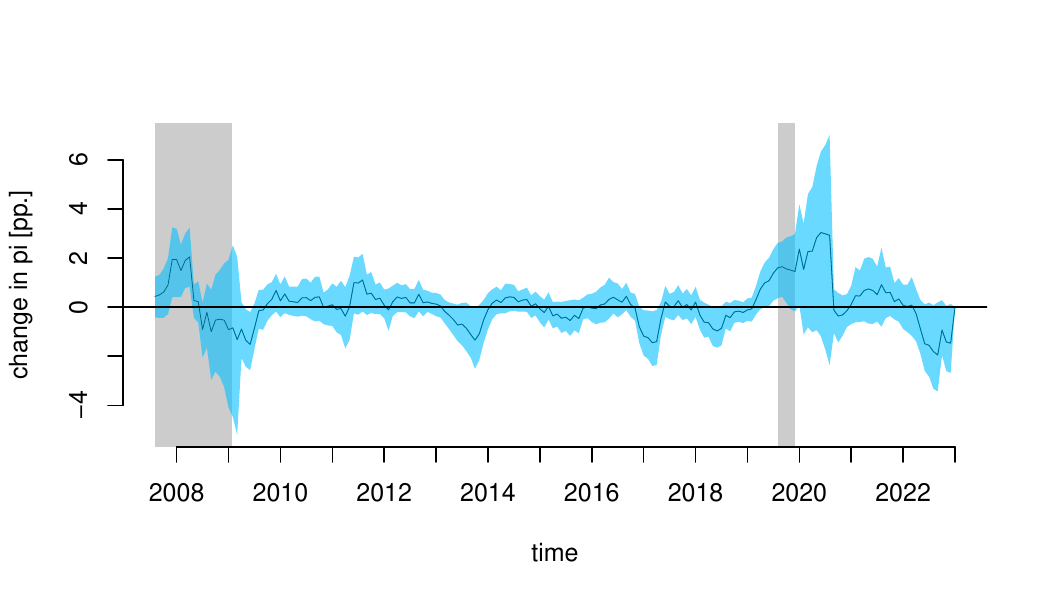}}  \\
		{\footnotesize Note: Figure shows the impulse responses of industrial production and inflation to monetary policy shocks cumulated over one year (solid lines) together with the 68\% highest posterior density set, starting in January 2008. The gray shaded areas mark the NBER recessions during the reported period. }
\end{figure}

Post GFC and COVID-19 crisis monetary policy had buffering effects. After the crisis, expansionary monetary policy lead to sharp increases in economic activity reaching 6.5\% in August 2009 and 10.8\% in January 2021. The effects are clearly different from zero then. Compared to these episodes, monetary policy shocks lead to minor changes in industrial production in the remaining periods confirming the impulse response analysis showing no effect for the first regime.

The cumulative responses of inflation to monetary policy shocks are positive and clearly different from zero in the aftermaths of the GFC and COVID-19 crisis. Mitigating the recessions comes at the cost of higher inflation in the short term. For a year, starting in mid-2009, monetary policy had a disinflationary effect, reaching up to -1.5 pp in November 2009. The effect seems not statistically different from zero, though. During the period starting in 2010 up to early 2020 inflation reacted little to recent monetary policy shocks. The sharp inflationary effect in late 2020 reached 3 pp. Finally, tightening of the monetary policy after 2022 was deflationary with the strongest cumulative response at the level of -2 pp in February 2023.


Overall, our model indicates that in the last fifteen years, US monetary policy targeted mainly sharp economic crises providing a cushioning effect on the economy. Their actions were strongly expansionary in the aftermaths of the global financial and COVID-19 crises. The stimuli came with a delay and put substantial inflationary pressure. In normal times, the central bank intervenes to a minor extent. 

\subsection{What if  the second regime never emerged?}

\noindent The interventions of the Fed were particularly essential after the GFC and COVID-19 crisis. Since systematic monetary policy predominately can be characterized by the TR augmented with money during these periods, we further investigate the time-varying effects of these policy actions. Specifically, we ask the question of how economic activity, inflation, short-term interest rates, term spread, money, and stock prices would have evolved if the monetary policy as conducted in the second regime never emerged. 

To that end, we generate a counterfactual scenario in which the economic developments propagate according to the monetary policy rule from the first regime and compare it to the actual values predicted by our unaltered model. Consequently, we first isolate standardized shocks deprived of the mean, volatility, and contemporaneous relationships and propagate them according to the parameter values where the third row of the structural matrix -- the monetary policy reaction function -- and the corresponding volatility of the log-volatility parameter include the values from the first regime only. Subsequently, we calculate counterfactual values of the economic outcomes.

Figure~\ref{fig:CFA} shows the counterfactual values in comparison to the the actual values of the time series for January 2008 until June 2023. If monetary policy as conducted in the second regime did not occur, inflation would be on average 1 pp higher after the GFC as reported in Figure~\ref{fig:CFA}(b). This result is in line with the stronger impulse response of inflation in the first regime combined with on average negative monetary policy shocks for the considered period. We also link this result to the missing inflation interpretation by \cite{bobeica2019}. In our model, if the second regime in the monetary policy reaction function did not emerge and inflation was forecasted using the parameter values and identification from the first regime, the model would predict a much higher inflation rate after 2008 than that which eventuated.

\begin{figure}[t!]
	\caption{Counterfactual and actual values of the time series} \label{fig:CFA}
	\subfloat[$y$]{\includegraphics[trim={0.0cm 1.0cm 0.0cm 2.0cm},width=0.5\textwidth]{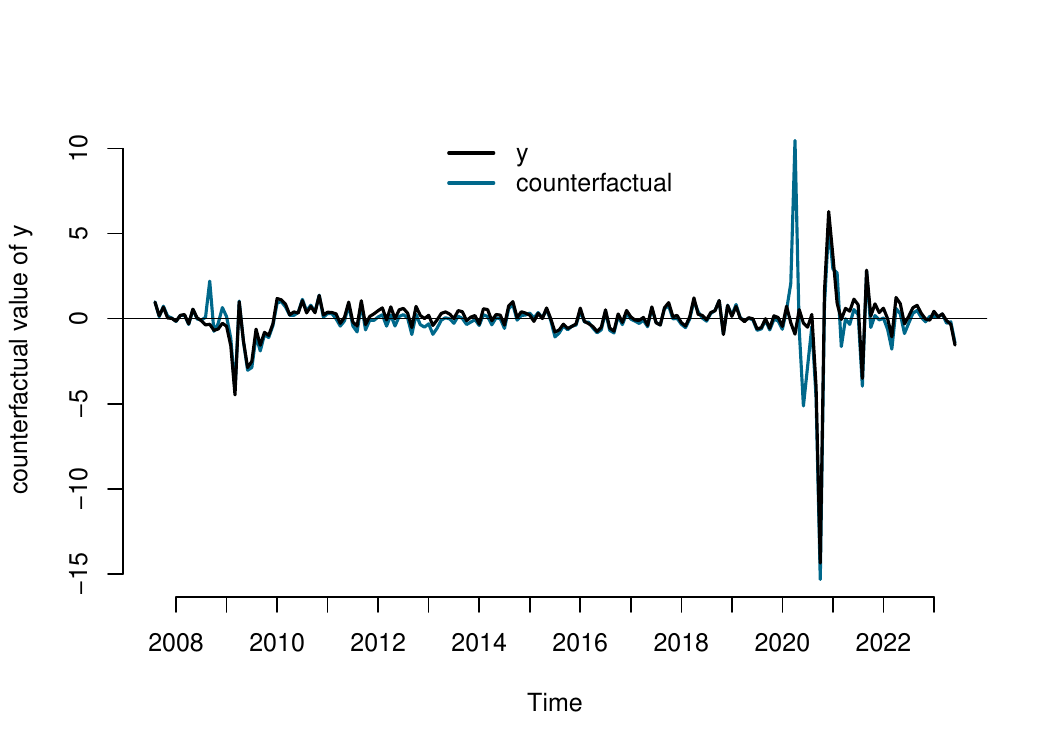}}
	\subfloat[$\pi$]{\includegraphics[trim={0.0cm 1.0cm 0.0cm 2.0cm},width=0.5\textwidth]{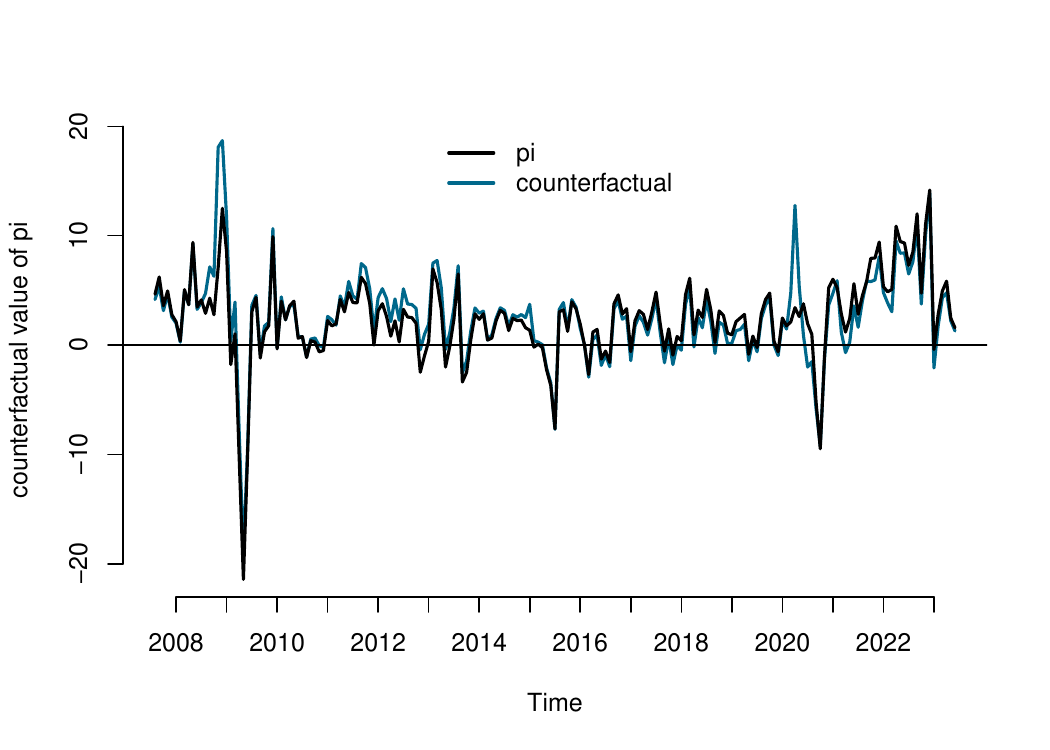}} \\
	\subfloat[$R$]{	\includegraphics[trim={0.0cm 1.0cm 0.0cm 2.0cm},width=0.5\textwidth]{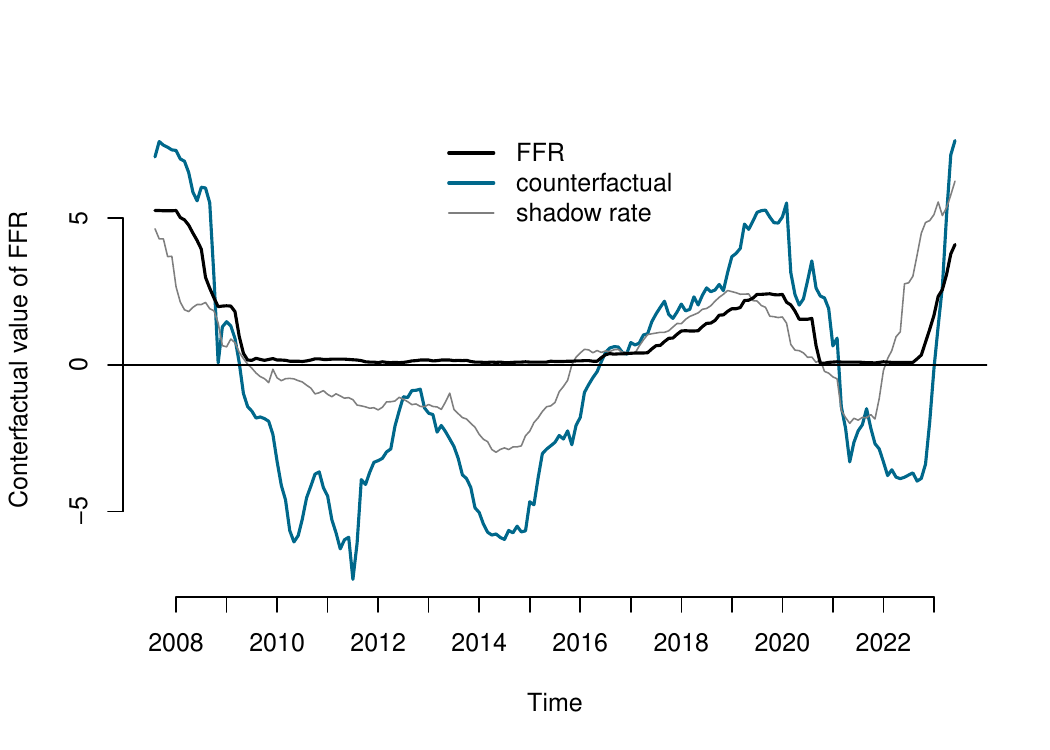}} 
	\subfloat[$TS$]{	\includegraphics[trim={0.0cm 1.0cm 0.0cm 2.0cm},width=0.5\textwidth]{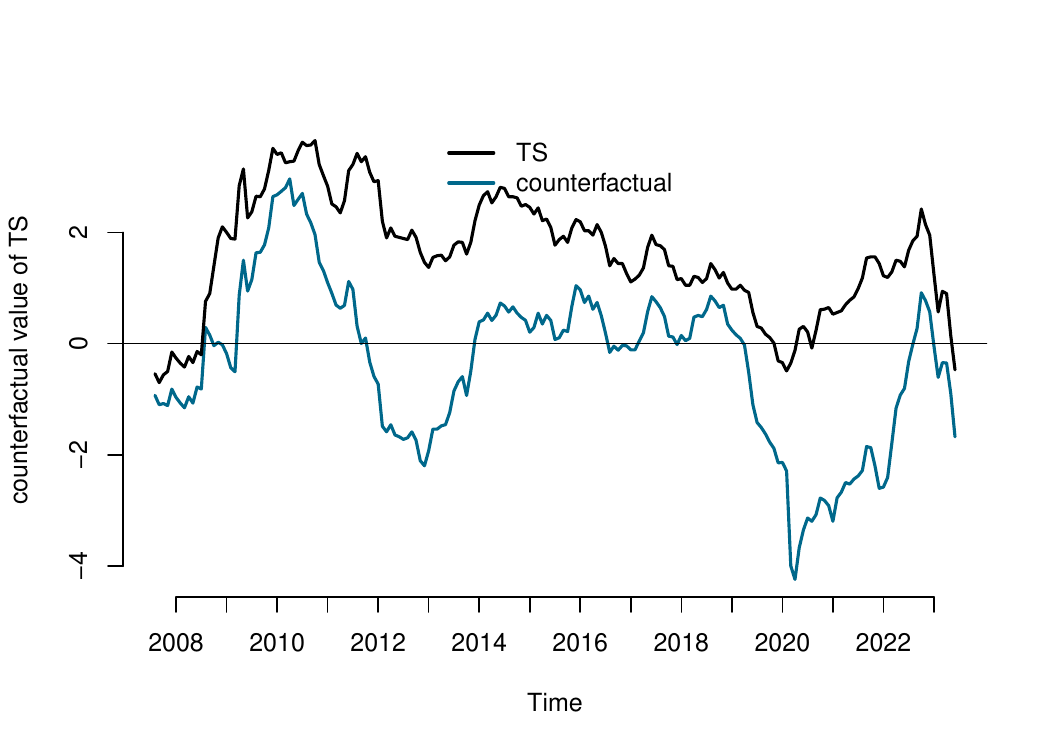}}\\
	\subfloat[$m$]{	\includegraphics[trim={0.0cm 1.0cm 0.0cm 2.0cm},width=0.5\textwidth]{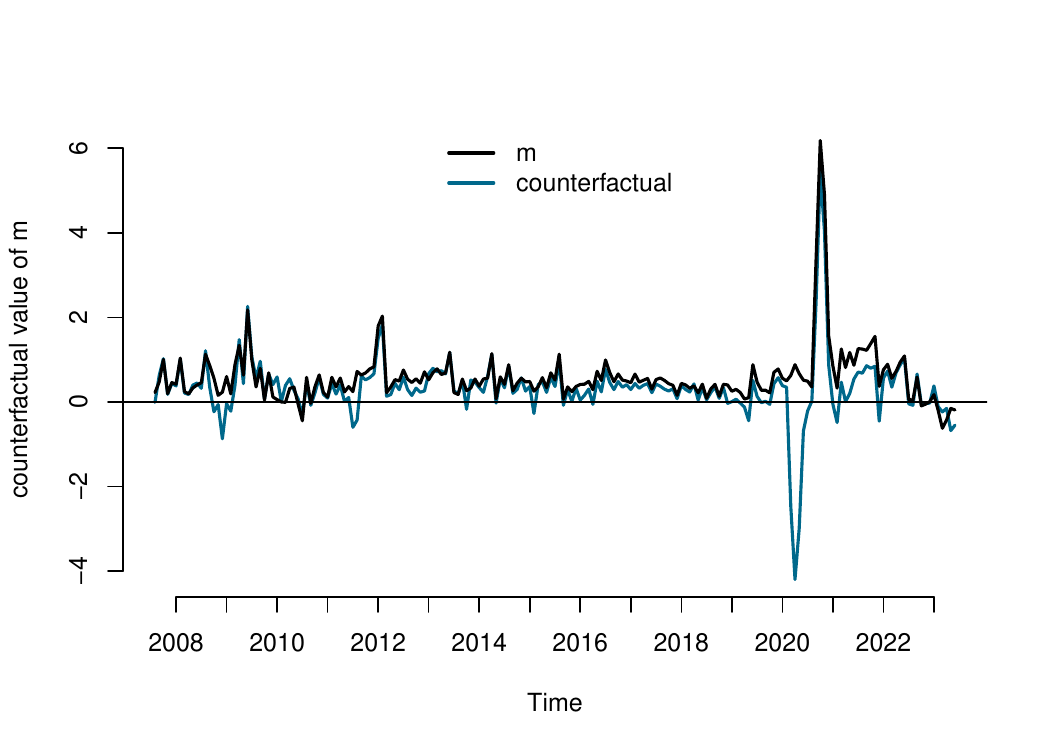}}
	\subfloat[$sp$]{	\includegraphics[trim={0.0cm 1.0cm 0.0cm 2.0cm},width=0.5\textwidth]{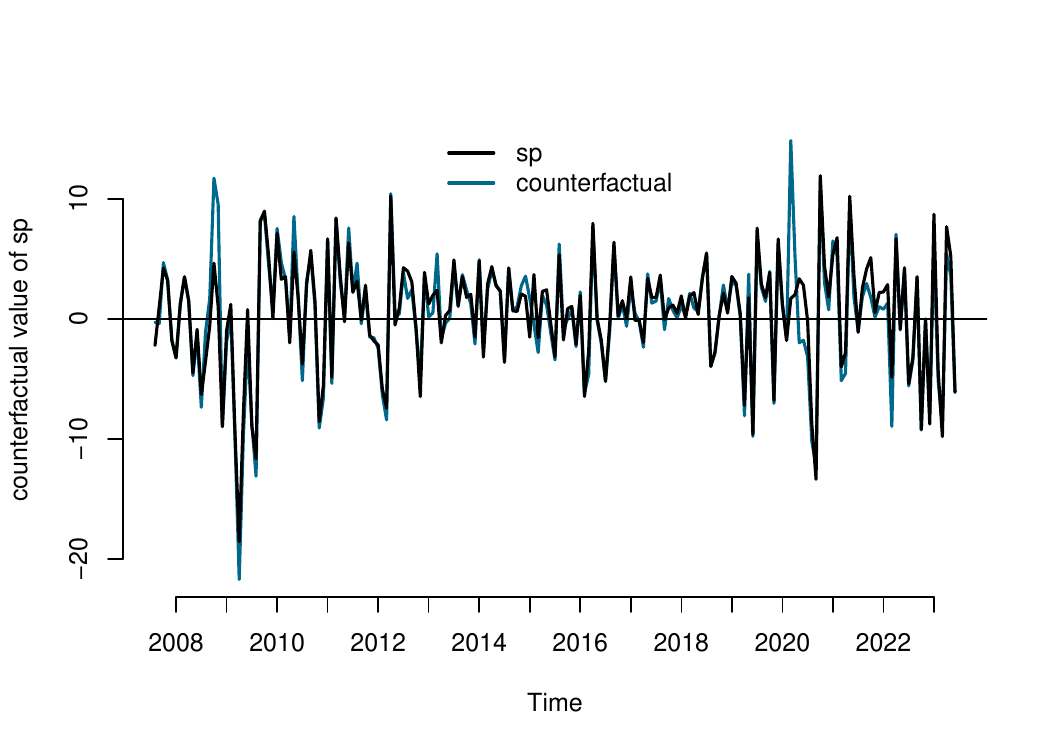}}
	\\	{\footnotesize Note: Figure shows the actual values of the six time series (dark lines) together with the counterfactual value (light lines) starting in January 2008. The analysis transforms $y_t$, $m_t$, and $sp_t$ to their growth rates.}
\end{figure}

Short term interest rates would be on average 1.4 pp lower. During zero lower bound periods including the COVID-19 pandemic, the counterfactual value of the federal funds rates is negative. Indeed, albeit more volatile, it follows closely the movements in the shadow rate represented by the gray line in Figure~\ref{fig:CFA}(c). The term spread would be considerably lower, especially from 2011 to 2016, on average by around 1.2 pp. A~similar result is reported by \cite{Baumeister2013} and, as the authors argue, it does not contradict the stabilizing effects of the monetary policy during this period. This is supported by the lower standard deviations of the term spread shock in regime 2 reported in Table~\ref{tab:regimeparam}. 

Money and industrial production growth rates as well as stock returns are not substantially different from the actual values except for specific episodes. For instance, in the second quarter of 2020, the counterfactual value of money growth is much lower and those of stock returns and industrial production growth are higher than the actual values. Money growth is also visibly lower in the counterfactual scenario throughout 2021. Summarizing, our counterfactual analysis provides evidence that monetary policy as conducted in the second regime was effective in lowering inflation, providing liquidity in the time of stress, and stabilizing the bond yield curve after 2008.

\section{Conclusions} \label{sec:conclusions}

\noindent We propose a new Bayesian structural vector autoregression which allows for time-varying identification. Our model facilitates a data-driven search of regime-specific exclusion restriction patterns. It features persistent Markov-switching regimes in the structural matrix and the volatility of the log-volatility as well as structural shock variances evolving at a much higher frequency following stochastic volatility. These components allow us to study time-variation in monetary policy identification, time-varying impulse responses, and provide an easy way to check identification via heteroskedasticity within each regime as a validity condition of regime-specific identification.

We find solid and robust support for time-variation in US monetary policy shock identification. Time-varying identification is essential to characterise unconventional monetary policy in the aftermath of the global financial and COVID-19 crises. Moreover, correctly specifying the model in the state-dependent setting enables sharp identification in each of the Markov-switching regimes, and provides wide interpretability of various reported outcomes.

Our results show that in the last fifteen years, monetary policy was conducted unconventionally and acted mainly as a buffer against the adverse economic effects of the financial crisis and the COVID-19 pandemic. The Fed provided stimuli to economic activity and liquidity to the market. At the same time, the central bank effectively curbs inflation  which in the absence of extraordinary means, would be on average one pp higher.

\bibliographystyle{elsarticle-harv}
\bibliography{bib}

\newpage
\appendix

\setstretch{1.5}
\section{Gibbs sampler} \label{sec-A:sampler}

\noindent For simplicity, we write the VAR model in compact form
\begin{align}
\mathbf{y}_t &= \mathbf{A}\mathbf{x}_{t}  + \boldsymbol{\varepsilon}_t
\end{align}
where the $N\times (Np+d)$-matrix  $\mathbf{A}=\begin{bmatrix}\mathbf{A}_1 &\dots & \mathbf{A}_p & \mathbf{A}_d\end{bmatrix}$ collects autoregressive and constant parameters and  the $(Np+d)$-vector $\mathbf{x}_t = \begin{bmatrix} \mathbf{y}_{t-1}' & \dots & \mathbf{y}_{t-p}' & \mathbf{d}_{t}'  \end{bmatrix}'$ contains lagged endogenous variables and constant terms. 

\subsection*{Sampling parameters in the structural matrix}

\noindent  To estimate the parameters in the structural matrix, we use a Gibbs sampler. The full conditional posterior distribution of the unrestricted elements in regime $m$, equation $n$, and TVI component $k$  is generalised-normal, proportional to
\begin{align}
	|\det\left( \mathbf{B}_{m.\mathbf{k}} \right)|^{T_m} \exp \left\{ -\frac{1}{2}
	 \mathbf{b}_{n.m.k} \mathbf{V}_{n.m.k} \overline{\boldsymbol{\Omega}}_{B.nm}^{-1} \mathbf{V}_{n.m.k}'\mathbf{b}_{n.m.k}'  \right\} \label{eq:postb}
\end{align} 
where $\mathbf{B}_{m.\mathbf{k}}$ is the matrix from regime $m$ and TVI component $k$, with scale matrix 
\begin{align}
	\overline{\boldsymbol{\Omega}}_{B.nm}^{-1} &= \gamma_{B.n}^{-1}\mathbf{I}_N + \sum_{t: s_t = m}
	(\mathbf{y}_t - \mathbf{A}\mathbf{x}_{t})(\mathbf{y}_t - \mathbf{A}\mathbf{x}_{t})'
\end{align}
where $T_m$ is the number of observations in regime $m$. Within each regime, we sample the structural matrix and the TVI indicators jointly.
\begin{enumerate}
	\item Sample parameters of the structural matrix: $\mathbf{B}$ row-by-row following the algorithm by \cite{WaggonerZha2003} by sampling $K$ vectors $\mathbf{b}_{n.m.k}$ from \eqref{eq:postb} given $\kappa(m)=k$ for all $k\in\{1,\dots,K\}$. 
	\item Sample TVI indicators: by computing the conditional posterior probabilities of each of these TVI components, $\overline{p}_{n.m.k}$,  proportional to the product of the likelihood function, $L(\theta\mid \mathbf{Y}_T)$, and the prior distributions:
	\begin{align}
		\overline{p}_{n.m.k} &\propto L(\theta\mid \mathbf{Y}_T,  \kappa(m)=k)p\left(\mathbf{B}_{n.m.k}\mid \gamma_B, \kappa(m)=k\right)p\left(\kappa(m) = k\right).
	\end{align}
	Sample the TVI indicator from the multinomial distribution using the probabilities $\overline{p}_{n.m.k}$ for $k=1,\dots,K$, and return the draw from the joint posterior $\left(\mathbf{b}_{n.m.k}, k\right)$.
	
	\item  Sample structural hyper-parameters: $\gamma_{B.n}$, $\underline{s}_{B.n}$, and $\underline{s}_{\gamma_B}$ from their respective full conditional posterior distributions:
	\begin{align}
		\gamma_{B.n}\mid \mathbf{B}_{n.m.k}, k, \underline{s}_{B.n} &\sim\mathcal{IG}2\left( \underline{s}_{B.n} + \sum_m \mathbf{b}_{n.m.k_n}\mathbf{b}_{n.m.k_n}', \quad\underline{\nu}_B + \sum_m r_{n.m.k_n} \right)\\
	\underline{s}_{B.n} \mid \gamma_B, \underline{s}_{\gamma_B} &\sim\mathcal{G}\left((\underline{s}_{\gamma_B}^{-1} + (2\gamma_{B.n})^{-1})^{-1},\underline{\nu}_{\gamma_B} + 0.5\underline{\nu}_{B}\right)\\
	\underline{s}_{\gamma_B} \mid \underline{s}_{B} &\sim\mathcal{IG}2\left(\underline{s}_{s_B} + 2\sum_{n=1}^N\underline{s}_{B.n}, \underline{\nu}_{s_B} + 2N\underline{\nu}_{\gamma_B}\right)
	\end{align}
\end{enumerate}
If TVI is not implemented in a particular row, sample $\mathbf{b}_{n.m}$ bypassing Step 2 of the algorithm.

\subsection*{Sampling of the remaining parameters}
\noindent Estimation of the remaining is implemented by the Gibbs sampler. Our emphasis was put on the selection of the most efficient sampling techniques that would facilitate estimation in larger systems of variables for which the dimension of the parameter space grows rapidly. 

The rows of the transition probabilities matrix, $\mathbf{P}$, and the initial state probabilities vector, $\boldsymbol{\pi}_0$, are all sampled from independent $M$-variate Dirichlet full conditional posterior distributions as in \cite{Fruhwirth-Schnatter2006}. The sampling procedure for the Markov process realisations follows the forward-filtering backward-sampling algorithm proposed by \cite{chib1996calculating}.



\subsubsection*{Sampling autoregressive parameters}

\noindent We follow \cite{ChanKoopYu2022} and sample the autoregressive parameters of the $\mathbf{A}$ matrix efficiently row-by-row, where the $n\textsuperscript{th}$ row is denoted by $[\mathbf{A}]_{n\cdot}$. Denote by $\mathbf{A}_{n=0}$ the $\mathbf{A}$ matrix with its $n\textsuperscript{th}$ row  set to zeros. To derive the sampler for the row of the autoregressive matrix, rewrite the structural form equation~\eqref{eq:sf} in an equivalent form as:
\begin{align}
\mathbf{B}_{m.k}\left(\mathbf{y}_t  - \mathbf{A}_{n=0}\mathbf{x}_{t}\right)&= \left([\mathbf{B}_{m.k}]_{\cdot n} \otimes \mathbf{x}_t'\right)[\mathbf{A}]'_{n\cdot}  + \boldsymbol{\varepsilon}_t \label{eq:sfsampleA}
\end{align}
Define an $n\times 1$ vector $\mathbf{z}_t^n = \mathbf{B}_{m.k}\left(\mathbf{y}_t  - \mathbf{A}_{n=0}\mathbf{x}_{t}\right)$ and an $N\times (Np+d)$ matrix $\mathbf{Z}_t^n = \left([\mathbf{B}_{m.k}]_{\cdot n} \otimes \mathbf{x}_t'\right)$ and rewrite equation~\eqref{eq:sfsampleA} as:
\begin{align}
\mathbf{z}_t^n&= \mathbf{Z}_t^n[\mathbf{A}]'_{n\cdot}  + \boldsymbol{\varepsilon}_t,
\end{align}
where the shocks follow the normal distribution defined in equation~\eqref{eq:sfshock}. This results in the multivariate normal full conditional posterior distribution for the rows of $\mathbf{A}$:
\begin{align}
[\mathbf{A}]'_{n\cdot} \mid \mathbf{Y}_T, [\mathbf{A}]_{n=0}, \mathbf{B}_{m.k},\boldsymbol{\sigma}^2_1, \dots\boldsymbol{\sigma}^2_T, \gamma_{A.n} &\sim \mathcal{N}_{Np+d}\left( \overline{\boldsymbol{\Omega}}_{A.n} \overline{\mathbf{m}}_{A.n}', \overline{\boldsymbol{\Omega}}_{A.n} \right)\\[1ex]
\overline{\boldsymbol{\Omega}}_{A.n}^{-1} &= \gamma_{A.n}^{-1}\underline{\boldsymbol{\Omega}}_A^{-1} + \sum_t \mathbf{Z}_t^{n\prime} \diag\left(\boldsymbol{\sigma}_t^2\right)^{-1}\mathbf{Z}_t^n \\
\overline{\mathbf{m}}_{A.n}'&= \gamma_{A.n}^{-1}\underline{\boldsymbol{\Omega}}_A^{-1}\underline{\mathbf{m}}_{A.n}' + \sum_t \mathbf{Z}_t^{n\prime} \diag\left(\boldsymbol{\sigma}_t^2\right)^{-1}\mathbf{z}_t^n 
\end{align}
Finally, the hyper-parameters of the autoregressive shrinkage hierarchical prior are sampled from their respective full conditional posterior distributions:
\begin{align}
\gamma_{A.n}\mid [\mathbf{A}]_{n\cdot}, \underline{s}_{A.n} &\sim\mathcal{IG}2\left( \underline{s}_{A.n} +  \left([\mathbf{A}]_{n\cdot} - \underline{\mathbf{m}}_{A.n}\right)\underline{\boldsymbol{\Omega}}_{A.n}^{-1}\left([\mathbf{A}]_{n\cdot} - \underline{\mathbf{m}}_{A.n}\right)', \quad\underline{\nu}_{A} + Np+d \right)\\
	\underline{s}_{A.n} \mid \gamma_{A.n}, \underline{s}_{\gamma_A} &\sim\mathcal{G}\left((\underline{s}_{\gamma_A}^{-1} + (2\gamma_{A.n})^{-1})^{-1},\underline{\nu}_{\gamma_A} + 0.5\underline{\nu}_{A}\right)\\
	\underline{s}_{\gamma_A} \mid \underline{s}_{A} &\sim\mathcal{IG}2\left(\underline{s}_{s_A} + 2\sum_{n=1}^N\underline{s}_{A.n}, \underline{\nu}_{s_A} + 2N\underline{\nu}_{\gamma_A}\right)
	\end{align}

\subsubsection*{Sampling stochastic volatility parameters} 

\noindent Gibbs sampler is facilitated using the auxiliary mixture sampler proposed by \cite{Omori2007}. 

Specify the $n\textsuperscript{th}$ structural shock as:
\begin{align}
u_{n.t} = \exp\left\{\frac{1}{2}\omega_{n}(s_t) h_{n.t}\right\} z_{n.t},
\end{align}
where $z_{n.t}$ is a standard normal innovation. Transform this equation by squaring and taking the logarithm of both sides obtaining:
\begin{align}\label{eq:loglin}
\widetilde{u}_{n.t} = \omega_{n}(s_t) h_{n.t} + \widetilde{z}_{n.t},
\end{align}
where $\widetilde{u}_{n.t} = \log u_{n.t}^2$ and $\widetilde{z}_{n.t}=\log z_{n.t}^2$. The distribution of $\widetilde{z}_{n.t}$ is $\log\chi^2_1$. This non-standard distribution is approximated precisely by a mixture of ten normal distributions defined by \cite{Omori2007}. This mixture of normals is specified by $q_{n.t}=1,\dots,10$ -- the mixture component indicator for the $n\textsuperscript{th}$ equation at time $t$, the normal component probability $\pi_{q_{n.t}}$, mean $\mu_{q_{n.t}}$, and variance $\sigma^2_{q_{n.t}}$. The latter three parameters are fixed and given in \cite{Omori2007}, while $q_{n.t}$ augments the parameter space and is estimated. Its prior distribution is multinomial with probabilities $\pi_{q_{n.t}}$. 

Define $T\times1$ vectors: $\mathbf{h}_n = \begin{bmatrix}h_{n.1}&\dots&h_{n.T}\end{bmatrix}'$ collecting the log-volatilities,  $\mathbf{q}_n=\begin{bmatrix}q_{n.1}&\dots&q_{n.T}\end{bmatrix}'$ collecting the realizations of $q_{n.t}$ for all $t$, $\boldsymbol\mu_{\mathbf{q}_n}=\begin{bmatrix}\mu_{q_{n.1}}&\dots&\mu_{q_{n.T}}\end{bmatrix}'$, and $\boldsymbol\sigma^2_{\mathbf{q}_n}=\begin{bmatrix}\sigma^2_{q_{n.1}}&\dots&\sigma^2_{q_{n.T}}\end{bmatrix}'$ collecting the $n\textsuperscript{th}$ equation auxiliary mixture means and variances, $\widetilde{\mathbf{U}}_n = \begin{bmatrix} \widetilde{u}_{n.1}&\dots&\widetilde{u}_{n.T} \end{bmatrix}'$, and $\boldsymbol\omega_{n} = \begin{bmatrix}\omega_{n}(s_1)&\dots& \omega_{n}(s_T)\end{bmatrix}'$ collecting the volatility of the volatility parameters according to their current time assignment based on the sampled realizations of the Markov process. Whenever a subscript on these vectors is extended by the Markov process' regime indicator $m$ it means that the vector contains only the $T_m$ observations for $t$ such that $s_t = m$, e.g., $\widetilde{\mathbf{U}}_{n.m}$. Finally, define a $T\times T$ matrix $\mathbf{H}_{\rho_n}$ with ones on the main diagonal, value $-\rho_n$ on the first subdiagonal, and zeros elsewhere, and a $T_m\times T_m$ matrix $\mathbf{H}_{\rho_n.m}$ with rows and columns selected according to the Markov state allocations such that $s_t = m$.
Sampling latent volatilities $\mathbf{h}_n$ proceeds independently for each $n$ from the following $T$-variate normal distribution parameterized following \cite{chib2006b} and \cite{ChanJeliazkov2009} in terms of its precision matrix $\overline{\mathbf{V}}^{-1}_{\mathbf{h}_n}$ and location vector $\overline{\mathbf{h}}_n$ as:

\begin{align}
\mathbf{h}_n \mid \mathbf{Y}_T, \mathbf{s}_n, \mathbf{q}_n, \mathbf{B}, \mathbf{A}, \boldsymbol\omega_n, \rho_n &
\sim\mathcal{N}_T\left( \overline{\mathbf{V}}_{\mathbf{h}_n}\overline{\mathbf{h}}_n,{} \overline{\mathbf{V}}_{\mathbf{h}_n}\right)\\
\overline{\mathbf{V}}_{\mathbf{h}_n}^{-1} &= \diag\left(\boldsymbol\omega_n^2\right)\diag\left(\boldsymbol\sigma^{-2}_{\mathbf{q}_n}\right) + \mathbf{H}_{\rho_n}'\mathbf{H}_{\rho_n} \\
\overline{\mathbf{h}}_n &= \diag\left(\boldsymbol\omega_n\right)\diag\left(\boldsymbol\sigma^{-2}_{\mathbf{q}_n}\right)\left(\widetilde{\mathbf{u}}_n - \boldsymbol\mu_{\mathbf{q}_n}\right)
\end{align}
The precision matrix, $\overline{\mathbf{V}}_{\mathbf{h}_n}^{-1}$, is tridiagonal, which greatly leads to a simulation smoother proposed by \cite{mccausland2011simulation}.

The regime-dependent volatility of the volatility parameters are sampled independently from the following normal distribution:
\begin{align}
\omega_n(m)\mid\mathbf{Y}_m, \mathbf{s}_{n.m}, \mathbf{s}_{n.m}, \mathbf{h}_{n.m}, \sigma_{\omega_n}^2 
&\sim\mathcal{N}\left(\overline{v}_{\omega_{n.m}}\overline{\omega}_{n.m},{} \overline{v}_{\omega_{n.m}}\right)\label{eq:postomega1}\\
\overline{v}_{\omega_{n.m}}^{-1} &= \mathbf{h}_{n.m}'\diag\left(\boldsymbol\sigma_{\mathbf{q}_n.m}^{-2}\right)\mathbf{h}_{n.m} + \sigma_{\omega_n}^{-2}\label{eq:postomega2} \\
\overline{\omega}_n &= \mathbf{h}_{n.m}'\diag\left(\boldsymbol\sigma_{\mathbf{q}_n.m}^{-2}\right)\left(\widetilde{\mathbf{u}}_{n.m} - \boldsymbol\mu_{\mathbf{q}_n.m}  \right)\label{eq:postomega3}
\end{align}

The autoregressive parameters of the SV equations are sampled independently from the truncated normal distribution using the algorithm proposed by \cite{robert1995simulation}:
\begin{align}
\rho_n\mid \mathbf{Y}_T, \mathbf{h}_n \sim\mathcal{N}\left( 
\left(\sum_{t=0}^{T-1}h_{n.t}^2 \right)^{-1} \left(\sum_{t=1}^{T}h_{n.t}h_{n.t-1} \right) ,{}\left(\sum_{t=0}^{T-1}h_{n.t}^2 \right)^{-1}\right)\mathcal{I}\left(|\rho_n|<1\right).\label{eq:samplerho}
\end{align}

The prior variances of parameters $\omega_{n}(m)$, $\sigma_{\omega_n}^2$, are \emph{a posteriori}  independent and sampled from the following generalized inverse Gaussian full conditional posterior distribution:
\begin{align}
\sigma_{\omega_n}^2\mid\mathbf{Y}_T, \omega_{n}(1), \dots, \omega_{n}(M) &\sim\mathcal{GIG}\left( M\underline{A}-\frac{1}{2},{} \sum_m \omega^2_{n}(m),{} \frac{2}{\underline{S}} \right)
\end{align}

The auxiliary mixture indicators $q_{n.t}$ are each sampled independently from a~multinomial distribution with the probabilities proportional to the product of the prior probabilities $\pi_{q_{n.t}}$ and the conditional likelihood function.

Finally, proceed to the ancillarity-sufficiency interweaving sampler proposed by \cite{Kastner2014}. Our implementation proceeds as follows: Having sampled random vector $\mathbf{\mathbf{h}}_n$ and parameters $\omega_{n}(m)$, compute the parameters of the centered parameterization $\tilde{h}_{n.t} = \omega_{n}(m) h_{n.t}$ and $\sigma_{\upsilon_n}^2=\omega^2_n(m)$. Then, sample $\sigma_{\upsilon_n.m}^2$ from the generalised inverse Gaussian full conditional posterior distribution:
\begin{align}
\sigma_{\upsilon_n.m}^2\mid\mathbf{Y}_m, \tilde{\mathbf{h}}_{n.m}, \sigma_{\omega_n}^2 \sim \mathcal{GIG}\left(-\frac{T_m-1}{2},{} \tilde{\mathbf{h}}_{n.m}'\mathbf{H}_{\rho_n.m}'\mathbf{H}_{\rho_n.m}\tilde{\mathbf{h}}_{n.m},{} \sigma_{\omega_n}^{-2}\right)
\end{align}
using the algorithm introduced by \cite{hormann2014generating}. Resample $\rho_n$ using the full conditional posterior distribution from equation~\eqref{eq:samplerho} where the vector $\mathbf{h}_n$ is replaced by $\tilde{\mathbf{h}}_n$. Finally, compute $\omega_{n}(m)=\pm\sqrt{\sigma_{\upsilon_n.m}^2}$ and $h_{n.t}=\frac{1}{\omega_{n}(m)}\tilde{h}_{n.t}$ and return them as the MCMC draws for these parameters.

\end{document}